\documentclass[lettersize,journal]{IEEEtran}
\usepackage{amsmath,amsfonts}
\usepackage{algorithm}
\usepackage{array}

\usepackage{textcomp}
\usepackage{stfloats}
\usepackage{url}
\usepackage{verbatim}
\usepackage{graphicx}
\usepackage{cite}
\usepackage{xcolor}

\usepackage{subcaption}
\usepackage{float}
\usepackage{enumitem}
\usepackage{booktabs}
\usepackage{multirow}
\usepackage{comment}
\usepackage{framed}
\usepackage{amssymb}
\usepackage{amsthm}

\usepackage{amsthm}
\usepackage{algpseudocode}

\newcommand{\mc}{\mathcal}
\newcommand{\mb}{\mathbf}

\newcommand{\argmin}{\mathop{\rm argmin}}
\newcommand{\argmax}{\mathop{\rm argmax}}
\newcommand{\prox}{\operatorname{\rm prox}}

\newtheorem{lemma}{Lemma}

\newtheorem{proposition}{Proposition}

\newtheorem{claim}{Claim}

\newcommand{\wang}[1]{{\color{black} #1}}
\newcommand{\sndrev}[1]{{\color{black} #1}}

\newcommand{\LineComment}[1]{\Statex \(\triangleright\) #1}

\usepackage[hidelinks,bookmarks=false]{hyperref}

\begin{document}

\title{Clock-Free Stochastic Scheduling for Real-Time \\ Synchronized Content Delivery}

\author{Xuan~Wang, Xiaoyi~Wu, Ying~Sun, Atilla~Eryilmaz,~\IEEEmembership{Senior Member,~IEEE,}
and Bin~Li,~\IEEEmembership{Senior Member,~IEEE}%
\thanks{
This work was supported in part by the ARO Grant W911NF-24-1-0103, NSF AI Institute (AI-EDGE) 2112471, CNS-NeTS-2106679, ECCS-CSCS-2619078.

Xuan Wang, Xiaoyi Wu, Ying Sun, and Bin Li are with the Department of Electrical Engineering,
The Pennsylvania State University, University Park, PA, USA. Atilla Eryilmaz is with the Department of
Electrical and Computer Engineering, The Ohio State University, Columbus, OH, USA.}%
}

\maketitle

\begin{abstract}
Real-time synchronized content delivery is critical for applications such as cloud gaming, collaborative mixed reality, and live streaming, where a server must ensure that a group of users perceives the same content simultaneously for each delivery.  Due to heterogeneous and unpredictable network delays, ensuring synchronization is highly challenging. We propose a clock-free scheduling framework that introduces latency compensation to control the server's delivery time for each user, aiming to maximize synchronization performance (i.e., minimizing the gap between the maximum and minimum latencies) while maintaining a low overall latency cost. Unlike prior work requiring synchronized clocks or client-side timestamps, our approach uses only acknowledgment feedback, enabling flexible deployment across diverse environments. We formulate the problem as a convex stochastic optimization with a non-smooth objective. To solve the problem, we employ the stochastic proximal point (SPP) algorithm that updates the delay compensation, using only the previous latency samples, rather than the latency distribution. By exploiting the subproblem structure, we apply ADMM to enable efficient updates with per-iteration complexity of at most $\mathcal{O}(N\log N)$, where $N$ is the number of users. We establish a convergence rate of $\mathcal{O}(1/\sqrt{K})$, where $K$ is the number of content deliveries for transmission. The framework is further extended to handle transmission failures, where packet loss results in successful delivery to only a subset of users. Finally, simulations using both synthetic and real-world latency traces demonstrate the consistent superiority of our algorithm against baseline methods.
\end{abstract}

\begin{IEEEkeywords}
Collaborative Content Delivery, Synchronized Scheduling, Stochastic Optimization.
\end{IEEEkeywords}

\section{Introduction}
\IEEEPARstart{W}{e} consider the problem of real-time synchronized scheduling for collaborative content delivery, motivated by applications such as cloud gaming \cite{huang2013gaminganywhere}, collaborative mixed reality \cite{ens2019revisiting}, e-commerce live streaming \cite{li2022livenet}, and Industrial Internet of things (IIoT) system\cite{jia2021passive}. For example, in cloud gaming, a cloud server renders high-quality images and must deliver them simultaneously to a group of players to ensure a smooth gaming experience. In collaborative mixed reality (MR) based firefighter training, the server runs fire accident simulations and delivers high-quality virtual fire visuals to the MR headsets of a group of firefighters depending on their location and perspective. It is essential that all firefighters see the same status of virtual fire accidents simultaneously to facilitate effective collaboration. In e-commerce live streaming, the server must ensure that all users see the same live status to avoid chaos during limited-time promotional offers. In autonomous driving, a high-speed vehicle must simultaneously broadcast its real-time statuses, such as location and speed, to neighboring vehicles for collision prevention and overall safe navigation. For IIoT systems, concurrent commands from the control center to distributed devices are crucial for cohesive collaboration. In all these applications, synchronized content reception for a group of users is crucial. However, fluctuating network transmission rates naturally cause asynchronous content reception for each user, leading to non-smooth collaborative experiences.

One possible solution is for the server to deliver the content with precise display times for users through an efficient network time protocol (NTP) \cite{mills2002internet} or precision time protocol (PTP) \cite{garg2017wireless}, which provides synchronized clock information. However, this strategy faces significant practical limitations. First, NTP may experience calibration errors of up to $100$ ms in dynamic network environments, while the high precision of PTP requires specialized hardware typically unavailable on consumer devices~\cite{watt2015understanding}. Moreover, the content in the above collaborative applications is far more complex than simple commands and text, involving substantial amounts of data. For instance, a high-resolution rendered gaming image can exceed 5MB\cite{illahi2020cloud}. In such a case, accurately estimating content delivery time is challenging due to fluctuating network transmission rates. Consequently, NTP and PTP, which merely provide synchronized clock information, fail to ensure synchronization among a group of users for each content delivery. Therefore, the efficient design of real-time synchronized scheduling for collaborative content delivery is essential. 

The majority of existing content delivery methods aim to minimize transmission latency. Some researchers use multipath routing to prioritize scheduling data delivery designed for large-scale networks with multiple hops \cite{ali2021packet}. In multi-user video streaming applications, most algorithms adaptively adjust the bitrate of images or the buffer size according to the alteration of latency, which is at the cost of image quality and thus leads to a fluctuating user experience \cite{gao2018multi}\cite{chen2022enhancing}\cite{huang2023buffer}. However, these methods do not address the synchronization problem, which is particularly critical for cloud gaming, collaborative MR, and e-commerce live streaming.

Many works (e.g., \cite{mills2002internet, garg2017wireless}) studied different types of clock synchronization in the field of network, such as NTP and PTP, but typically these kinds of synchronization are achieved through frequent calibration of the target clock by minimizing the observed clock deviation from the time reference. Subsequently, some immersive applications (e.g., \cite{simiscuka2019real}) began addressing the synchronization problem, as network synchronization significantly impacts the consistency of views in a shared scene, thereby influencing the users' sense of presence. However, most of these applications still rely on the classical NTP approach. Although existing studies (e.g., \cite{hamza2004adaptive,schnierle2023latency}) employed network latency compensation to ensure user synchronization, they used heuristic methods to manipulate stochastic network latency to achieve synchronization, rather than providing a theoretical guarantee of their approach. 

In this paper, we propose a clock-free scheduling framework that introduces latency compensation to control the server's delivery time for each user, enhancing synchronization performance by minimizing the gap between maximum and minimum latencies. \wang{The framework uses delivery feedback from users to estimate user-specific latency, allowing synchronization to be handled at the server side with minimal client-side timing requirements.} However, this approach also increases the latency for each user to receive and view the content. Our goal is to determine the optimal latency compensation that maximizes synchronization performance while keeping the overall latency cost as low as possible. Due to the randomness of latency, our formulated problem is a stochastic optimization problem. Given that we can collect latency samples from each user based on previous content deliveries, our aim is to develop an efficient, low-complexity stochastic optimization method with fast convergence. However, our objective function includes a non-smooth component arising from the difference between the maximum and minimum latencies after latency compensation. To address this, we employ the stochastic proximal point (SPP) method for each content delivery using the previous latency samples collected from each user, leading to non-smooth deterministic subproblems. Then, instead of directly using a general-purposed solver such as SDPT3 and SeDuMi, we explore the inherent problem structure and utilize the alternating direction method of multipliers (ADMM) to decompose the computation, resulting in a computational complexity of at most $\mathcal{O}(N\log N)$ per iteration. 
We further demonstrate that our proposed algorithm achieves a convergence rate of $\mathcal{O}(1/\sqrt{K})$ and its adaptability in handling packet drops in realistic network environments. Our main contributions are summarized as follows:

$\bullet$ We formulate the real-time synchronized scheduling problem as a stochastic convex optimization task, aiming at minimizing both the maximum-minimum latency gap and the overall latency cost by controlling latency compensation. Unlike prior work relying on synchronized clocks, our approach is clock-free and uses only acknowledgment feedback \wang{naturally available in systems such as cloud gaming and MR, requiring no specialized hardware or client-side timestamps}, enabling practical deployment in diverse network environments.

$\bullet$ We propose the Synchronized Content Scheduling (SCS) algorithm based on the SPP method, which efficiently uses only previous latency samples. By exploiting the subproblem structure and applying ADMM to decompose computations, the algorithm achieves a lightweight $\mathcal{O}(N\log N)$ scalar operations per iteration. We further extend the framework to handle packet-dropping scenarios by updating only the subset of acknowledged users, thereby ensuring robust synchronization under unreliable network conditions.

$\bullet$ We establish a convergence rate of $\mathcal{O}(1/\sqrt{K})$ for the proposed algorithm, and validate its performance through extensive simulations using both synthetic and real-world latency traces. Extensive experimental results demonstrate the superiority of the proposed algorithm over baseline methods in synchronization performance and computational efficiency.

\emph{Note on Notation:} We use boldface to denote vectors and script font to denote sets. The symbol $\succeq$ represents the component-wise ``greater than or equal to" operation. The $\ell_2$ norm of a vector $\mb{x}$ is denoted by $\|\mb{x}\|$. 
Let $\mc{H}$ be an index set, then we denote $\|\mb{x}\|_{\mc{H}}$ as the $\ell_2$ norm restricted to the set $\mc{H}$, i.e. $\|\mb{x}\|_{\mc{H}}^2 \triangleq \sum_{n\in \mc{H}}(x(n))^2$.
Let $\overline{\mb{x}}$ denote the vector of the mean of the elements in the vector $\mb{x}$. We use $[N]$ to denote the set of $\{1,2,\ldots,N\}$. Let $\mathbb{R}^N$ and $\mathbb{Z}_{+}$ denote the $N-$dimensional Euclidean space and the set of positive integers, respectively. 
We use $\mathcal{O}(\cdot)$ and $\Omega(\cdot)$ to denote asymptotic upper and lower bounds, respectively. Specifically, for functions $f(x)$ and $g(x) > 0$, we denote $f(x) = \mathcal{O}(g(x))$ if there exist constants $C > 0$ and $x_0$ such that $|f(x)| \leq C g(x)$ for all $x \geq x_0$, and $f(x) = \Omega(g(x))$ if there exist constants $C > 0$ and $x_0$ such that $|f(x)| \geq C g(x)$ for all $x \geq x_0$.

\section{Related Work}
\label{sec:related_work}

In this section, we provide an overview of key topics
closely related to our work: collaborative content delivery, synchronization in networking, and stochastic network scheduling.

\textbf{Collaborative content delivery:} Collaborative content delivery underpins a wide range of applications, such as cloud gaming, video streaming, and VR, all of which require tight inter-user synchronization to maintain quality of experience (QoE). In online games, for instance, latencies exceeding 20 to 30ms can significantly degrade the user experience~\cite{stauffert2020latency}. \textit{Li et al.} \cite{li2018lag} propose a real-time lag compensation mechanism for cloud gaming that equalizes uplink and downlink delays, improving fairness and QoE in multiplayer FPS games. \textit{Kassir et al.}\cite{kassir2021joint} propose a joint update rate adaptation algorithm for multiplayer cloud gaming that dynamically adjusts player update rates to mitigate the impact of network delays. In addition, real-time content delivery also plays a critical role in Video streaming and VR applications.
\textit{Gao et al.}\cite{gao2018multi} propose a cooperative mobile video streaming framework that optimizes bitrate adaptation and downloading using a Lyapunov-based algorithm to enhance QoE in multi-user environments. \textit{Park et al.} \cite{park2018minimizing} optimize end-to-end latency in wireless VR social interactions by allocating bandwidth, thereby addressing synchronization challenges arising from disparities between virtual and physical users. 

While these studies offer valuable strategies for improving QoE, like reducing latency and enhancing video quality, they mainly target individual user performance and often overlook the stringent synchronization demands for simultaneous content reception across users. Even when synchronization is considered, most approaches rely on heuristic clock alignment without theoretical guarantees on synchronization performance. This critical gap highlights the need for provably efficient and theoretically grounded algorithms for real-time synchronization in collaborative content delivery applications.

\textbf{Synchronization in networking:} Clock synchronization is crucial for accurate timekeeping in distributed systems. Numerous studies (e.g., \cite{mills2002internet, garg2017wireless}) have explored various synchronization techniques such as NTP and PTP, which operate by periodically calibrating local clocks against a reference source. Software-based NTP is vulnerable to path asymmetry and jitter, which can result in calibration errors of up to 100 ms in wireless or mobile environments. PTP improves accuracy through hardware time-stamping at network interface controllers, but it requires specialized hardware that is typically absent in commodity systems. Other techniques aim to enhance accuracy and robustness. Synchronous Ethernet (SyncE) \cite{ferrant2008synchronous}, for example, embeds timing information at the physical layer to reduce packet delay variation, but demands SyncE-compatible hardware with built-in clock recovery, making it incompatible with standard Ethernet infrastructure. Overall, existing methods either struggle under dynamic network conditions or rely on dedicated hardware for precision. To address these challenges, we propose a latency compensation technique that is efficient, easy to deploy, and robust against varying network conditions, offering a practical solution to high-precision synchronization.

\textbf{Stochastic network scheduling:}  Stochastic network scheduling has been extensively studied. \textit{Jiang and Walrand}\cite{jiang2009stable} propose the Deficit Maximum Weight algorithm for utility-maximizing and stable scheduling in stochastic processing networks. More recent research tailors scheduling to emerging applications and metrics. \textit{Li et al.}\cite{li2018efficient} develop an algorithm that equalizes packet completion times across multiple nodes. In distributed learning, \textit{Yang et al.}\cite{yang2019scheduling} analyze how wireless scheduling impacts the convergence rate of federated learning. While the focus of stochastic network scheduling is shifting from traditional metrics, such as throughput and delay, towards data freshness with the Age of Information \cite{kadota2018scheduling, joo2018wireless, zhou2023exploring, zhou2025novel}, the ``synchrony" in these pioneering works is limited to aligning data freshness. These approaches fail to address the stricter temporal alignment of content delivery essential for real-time applications, leaving application-aware synchronization largely under-explored.

\section{System Model}
\label{sec:model}
We consider the synchronized content transmission from a single server to $N$ users, whereby the server aims to deliver consecutive content to all users simultaneously so that they receive and view the content at the same time. Let $d_k(n)$ denote the latency experienced by user $n$ to receive and view the $k^{th}$ content. 
In practice, $d_k(n)$ is measured as the elapsed time between sending the content packet and receiving the corresponding acknowledgment (ACK). 
\sndrev{Since the ACK (a few bytes) is much smaller than the content payload, this RTT-like measurement serves as a practical proxy for the downlink delivery latency, although its accuracy may be affected by uplink/downlink asymmetry and ACK jitter.
We model $d_k(n)$ as a random variable to account for the fluctuating transmission rates, varying content information, and random ACK-path variations.} 
For example, in video streaming, the frame size changes over time, as shown in Fig. 2 in  \cite{xiong2022index}. We assume $d_k(n)$ is within a finite duration $\tau_{\max}$. For example, the finite duration (or maximum tolerable latency) in VR streaming is typically within $20$ms \cite{stauffert2020latency}. We assume that $\{d_k(n), \forall n=1,2,\ldots,N,\forall k=1,2,\ldots\}$ are identically and independently distributed (i.i.d.) across delivered content for each user and independently distributed across users\footnote{We note that this assumption aids in the development and analysis of our algorithm. However, as demonstrated through simulations with real-world data in Section \ref{sec:sim}, our proposed algorithm does not depend on this assumption.}. 
For each content $k$, the latency experienced by each user varies. Thus, we employ \emph{latency compensation} for each user to postpone their transmission and ensure synchronization across all users for each content delivery. Let $x_k(n)$ denote the latency compensation of user $n$ for delivering the $k^{th}$ content, i.e., the server delivers the $k^{th}$ content to user $n$ after $x_k(n)$ time units. Here, we note that $x_k(n)\geq0,\forall n,\forall k$, since we can only delay the transmission, not send it in advance. Let $\mb{d}_k\triangleq(d_k(n))_{n=1}^{N}$ and $\mb{x}_k\triangleq(x_k(n))_{n=1}^{N}$ represent the vectors for each user's latency and latency compensation 
when delivering the $k^{th}$ content, respectively.  

\begin{figure}[!t] 
    \centering
    \begin{subfigure}[t]{0.24\textwidth}  %
        \centering
        \includegraphics[width=\textwidth]{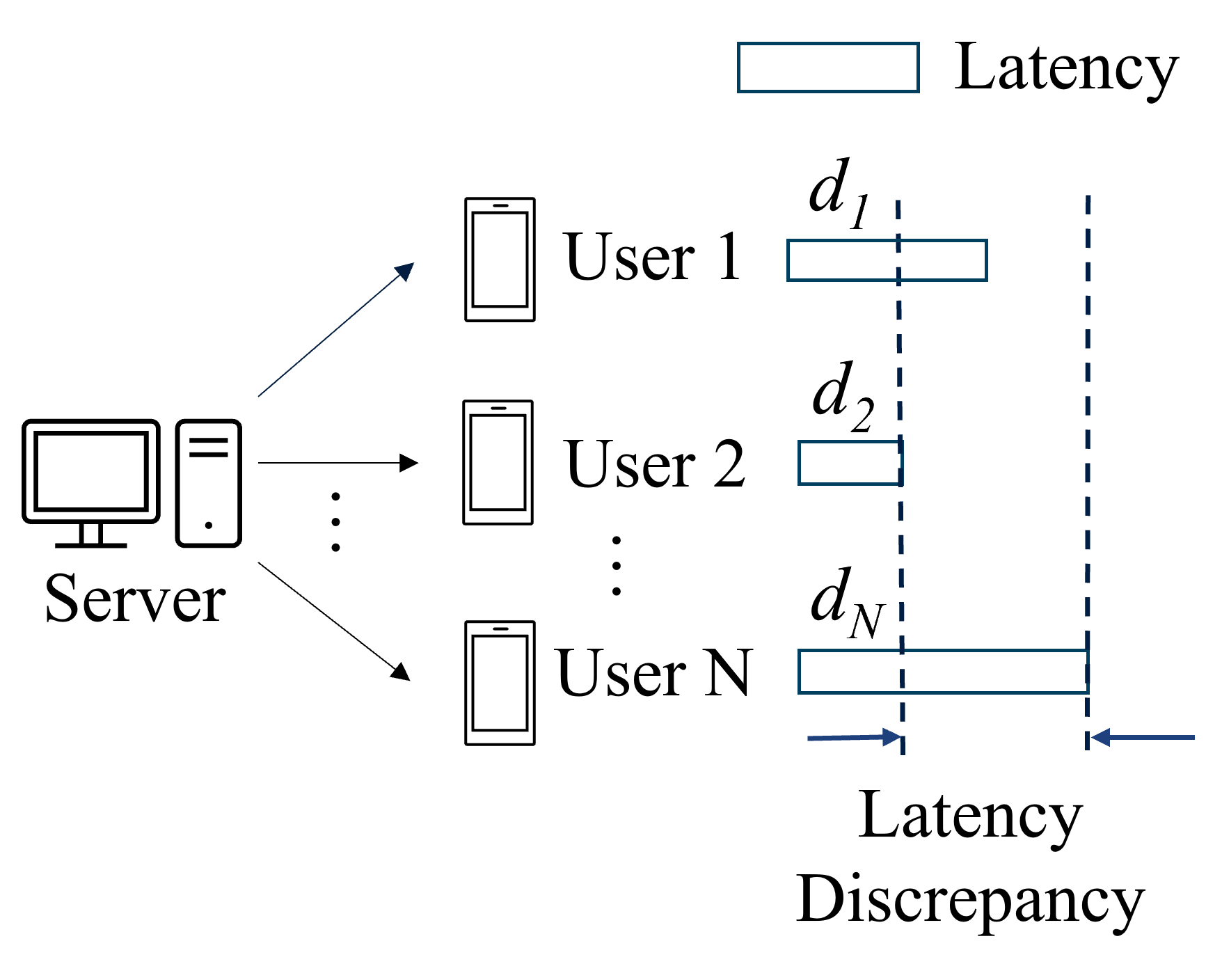}
        \caption{Before latency compensation}
        \label{fig:model:before}
    \end{subfigure}
    \begin{subfigure}[t]{0.24\textwidth}
        \centering
        \includegraphics[width=\textwidth]{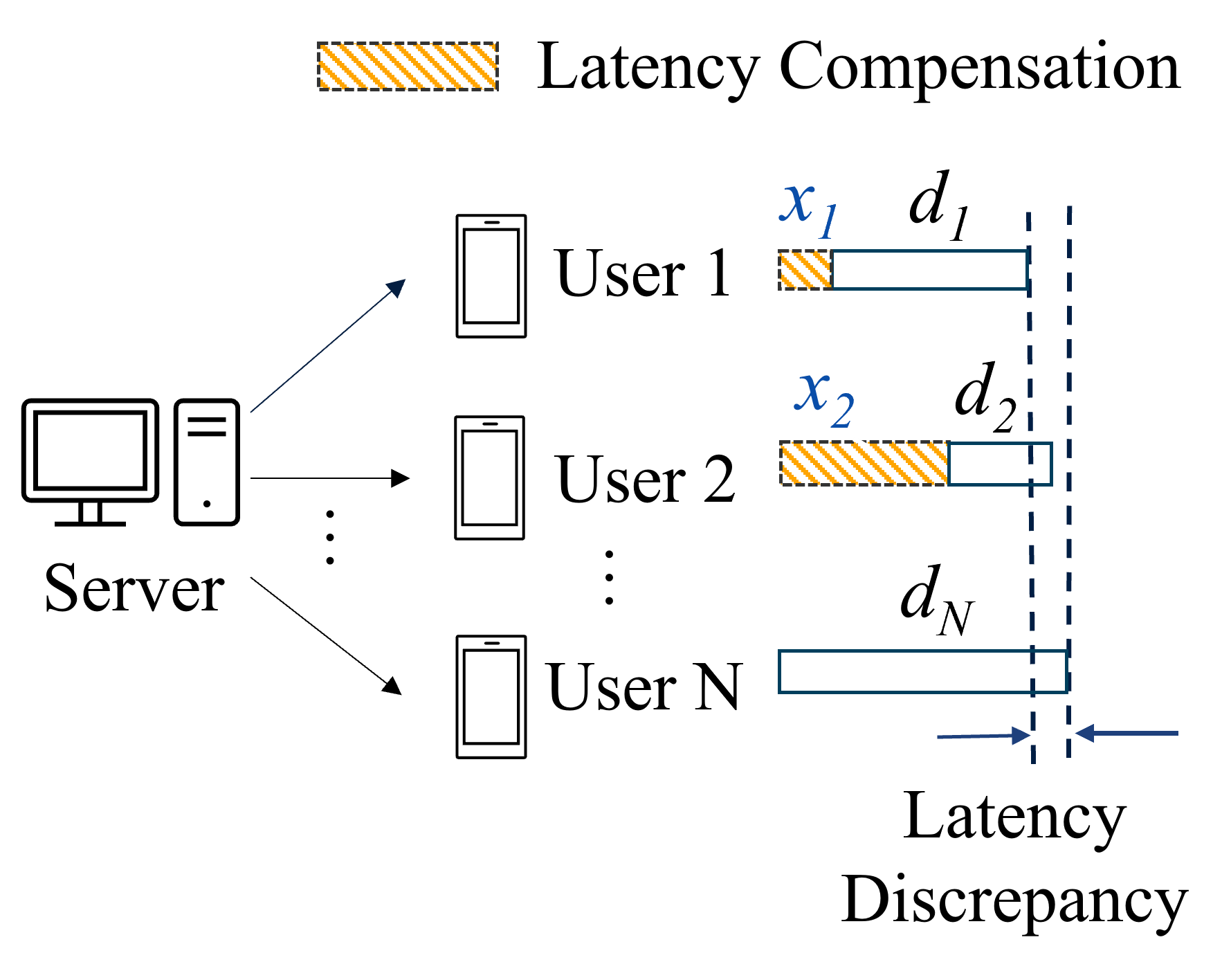}
        \caption{After latency compensation}
        \label{fig:model:after}
    \end{subfigure}
    \caption{Content delivery: incorporating latency compensation enhances synchronization for all users, though it increases latency for individual users.}
    \label{fig:model}
\end{figure}

While latency compensation can potentially improve synchronization among all users, it increases the latency for each user to receive and view the content. As such, we need to carefully manage the latency compensation to maximize synchronization among all users for each content delivery, i.e., minimizing the gap between maximum and minimum latencies after latency compensation while keeping the \wang{overall user-perceived latency} as small as possible. Fig. \ref{fig:model} illustrates the scenarios before and after applying latency compensation. It is evident from Fig. \ref{fig:model} that adding latency compensation enhances synchronization among all users, albeit at the cost of increased latency for each user. Since $\mb{d}_k$ is i.i.d. across delivered content, we omit the index $k$ for simplicity. As such, we aim to solve the following optimization problem: \begin{align} 
\label{eqn:obj}
\min_{\mathbf{x} \succeq \mb{0}} F(\mb{x})\triangleq\mathbb{E}_{\mathbf{d}} \left[ f(\mb{x};\mb{d})\right],
\end{align} 
where the instantaneous synchronization performance 
\begin{multline*}
    f(\mb{x};\mb{d})\triangleq \max_{n} [x(n) + d(n)] - \min_{n } [x(n) + d(n)] \\ + w \sum_{n=1}^{N}(c_n(x(n)+d(n))),
\end{multline*}
and $w\geq0$ is some non-negative weighting parameter that balances the synchronization and delay performance. The cost function $c_n(\cdot): \mathbb{R}_+ \rightarrow \mathbb{R}_+$  represents the latency-induced cost perceived by user $n$. It is assumed to be smooth, non-negative, and convex\footnote{A function $f$ is called \emph{convex} if $f(\lambda\mb{x} + (1-\lambda)\mb{y}) \leq \lambda f(\mb{x}) + (1-\lambda) f(\mb{y})$ holds for any vectors $\mb{x}$ and $\mb{y}$ within its domain and any parameter $\lambda \in [0, 1]$.}, reflecting increasing user dissatisfaction with higher latency. Typical examples of such cost functions include linear and quadratic forms such as $ c_n(x(n) + d(n)) = x(n) + d(n) $ or $ c_n(x(n) + d(n)) = (x(n) + d(n))^2 $, as discussed in~\cite{tsai2023distribution}. We can show that both $f( \, \cdot \,;\mb{d})$ and $F$ are proper\footnote{A function $f$ is called \emph{proper} if it never takes the value $-\infty$ and has at least one finite value.} and convex in $\mb{x}$, as demonstrated in Appendix \ref{appendix:convex}. This ensures the existence of the optimal solution to the optimization problem \eqref{eqn:obj}. Let $\mb{x}^\star \triangleq (x^\star(n))_{n=1}^N$ be the optimal solution to the above convex optimization problem, i.e., $\mb{x}^\star\in\argmin_{\mb{x}\succeq \mb{0}}F(\mb{x})$.   

Directly solving the optimization problem \eqref{eqn:obj} requires knowledge of the latency distributions for all users, which is typically unavailable in practice. Instead, prior to each content delivery, we can collect latency samples from each user based on the previous content delivery. This motivates us to develop an efficient stochastic optimization method with fast convergence. 
Here, we emphasize that synchronization performance is measured by the discrepancy between the maximum and minimum latencies after latency compensation.
\wang{This max-min criterion directly controls the worst inter-user latency gap, unlike variance-type metrics which measure average dispersion around the mean. For example, if one user has compensated latency $x(1) + d(1)= M$ and all others have $x(n) + d(n) =0$, then the max-min discrepancy is $M$, while the variance $M^2(N-1)/N^2$ decreases as the number of users increases. }
However, our stricter synchronization requirement renders the objective function non-smooth\footnote{A function is \emph{non-smooth} if it includes non-differentiable components, such as the maximum and minimum operations in our objective function.}, presenting considerable challenges in designing fast algorithms. Despite this, we explore the inherent problem structure to accelerate the computation, enabling real-time scheduling in practice. Next, we provide a numerical example demonstrating the benefits of incorporating latency compensation. 

\section{A Motivating Example}
\label{sec:example}

\begin{figure}[!t]  %
    \centering
    \begin{subfigure}[b]{0.24\textwidth}
        \centering
        \includegraphics[width=\textwidth]{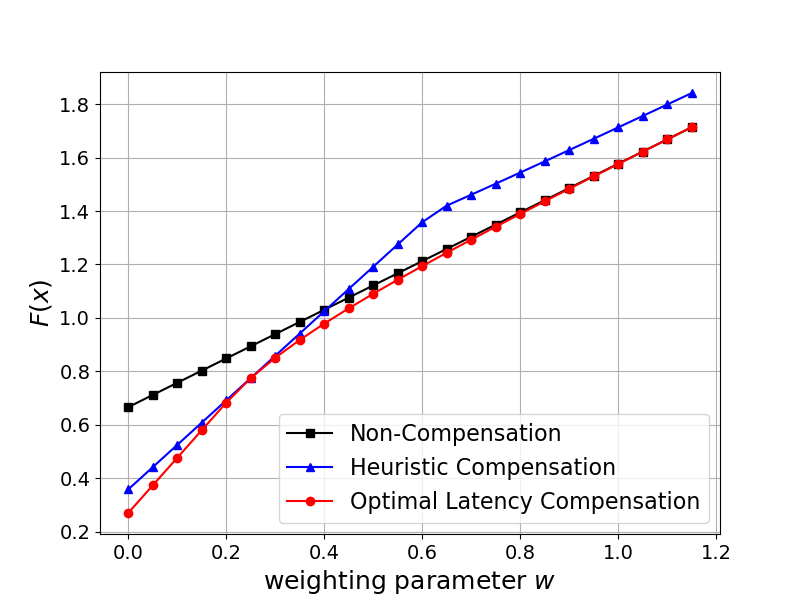}
        \caption{Objective value} %
        \label{fig:example:a}
    \end{subfigure}
    \hfill
    \begin{subfigure}[b]{0.24\textwidth}
        \centering
        \includegraphics[width=\textwidth]{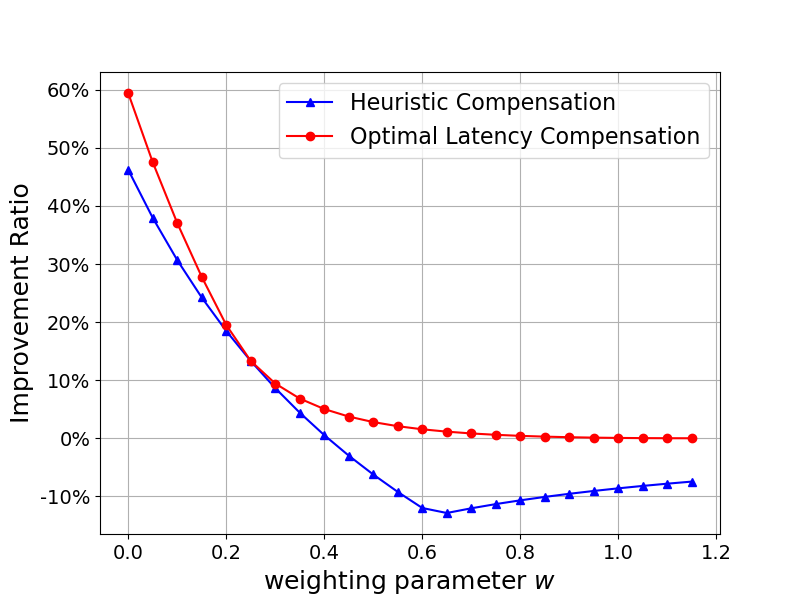}
        \caption{Improvement ratio}
        \label{fig:example:b}
    \end{subfigure}
    \caption{Impact of weighting parameter $w$.}
    \label{fig:example:both}
\end{figure}

In this section, we consider a simple (albeit hypothetical) numerical example to illustrate the advantages of incorporating latency compensation, given the latency distributions of all users.  Specifically, we consider a scenario with two users, where the latency of delivering content to each user follows the distribution described below. 
\begin{align*}
\mathbb{P}[d(n) = 1] = p_n \text{ and } \mathbb{P}[d(n) = 0.1] =  q_n, \forall n =1, 2,   
\end{align*}
where $0\leq p_n\leq 1, \forall n=1,2$, and $p_n + q_n = 1$. The cost function is defined as $c_n(x(n)+d(n)) \triangleq (x(n)+d(n))^2$. 
Hence, the objective function  in \eqref{eqn:obj} can be expressed as follows:
\begin{equation}
    \label{eqn:example}
    \begin{aligned}
    F(\mathbf{x}) &= p_1 p_2 \left( |x_1 - x_2| + w \left( (x_1 + 1)^2 + (x_2 + 1)^2 \right) \right) \\
    +& p_1 q_2 \left( |x_1 - x_2 + 0.9| + w \left( (x_1 + 1)^2 + (x_2 + 0.1)^2 \right) \right) \\
    +& q_1 p_2 \left( |x_1 - x_2 - 0.9| + w \left( (x_1 + 0.1)^2 + (x_2 + 1)^2 \right) \right) \\
    +& q_1q_2 \left( |x_1 - x_2| + w \left( (x_1 + 0.1)^2 + (x_2 + 0.1)^2 \right) \right).
\end{aligned}
\end{equation}

We can obtain the optimal latency compensation based on parameters $p_1, p_2$, and $w$ using an optimization solver\wang{, which serves as the benchmark.}
\wang{We compare it with a heuristic that aligns users according to their average compensated delivery times while retaining the latency-cost term. Specifically, it solves the deterministic problem:
\begin{align}
    \label{eqn:heuristic}
    \min_{\mb{x}\succeq \mb{0}} f(\mb{x};\mathbb{E}[\mb{d}]).
\end{align}
Here, $\mathbb{E}[\mb{d}]$ can be estimated from observed latency samples.}

To evaluate the synchronization performance, we define the \textit{improvement ratio} as:
\begin{equation}
    \text{Improvement Ratio}(\mb{x}) = \frac{F(\mb{0}) - F(\mb{x}) }{F(\mb{0})},
    \label{eqn:improveratio}
\end{equation}
where $F(\mb{x})$ is the objective function value under a given latency compensation vector $\mb{x}$, and $F(\mb{0})$ denotes the baseline value with no latency compensation applied. This metric measures the relative reduction in our synchronization objective function, $F(\cdot)$, compared to a baseline with no compensation. Therefore, a higher improvement ratio indicates a more significant enhancement in synchronization performance. Fig. 2 illustrates the performance of different latency compensation methods (i.e., optimal latency compensation, heuristic latency compensation, and no latency compensation) for $p_1=0.1 \;(q_1=0.9)$ and $p_2 = 0.8 \;(q_2=0.2)$.
We can see from Fig. 2a that the optimal latency compensation method results in the lowest objective function value. In contrast, the heuristic latency compensation performs even worse than no latency compensation when the weighting parameter $w$ is relatively large (e.g., $w=0.4$). Fig. \ref{fig:example:b} further demonstrates the impact of the weighting parameter $w$ on the improvement ratio compared to no latency compensation. We can observe from Fig. \ref{fig:example:b} that the improvement ratio of optimal latency compensation can reach up to around 60\%, outperforming the heuristic approach, particularly when the weighting parameter is small (e.g., less than 0.2). Moreover, the improvement ratio decreases as the weighting parameter $w$ increases. This is because a large $w$ places more emphasis on latency, and adding the latency compensation increases the latency for each user. Fig. \ref{fig:example:disparity} illustrates the impact of latency heterogeneity when the weighting parameter $w=0.2$. Fig. \ref{fig:example:disparity} shows that the improvement ratio is substantial, reaching up to 63.7\%, particularly when there is a significant disparity in user latency (e.g., a large difference between $p_1$ and $p_2$).

\begin{figure}[!t]
    \centering
    \includegraphics[width=0.7\linewidth]{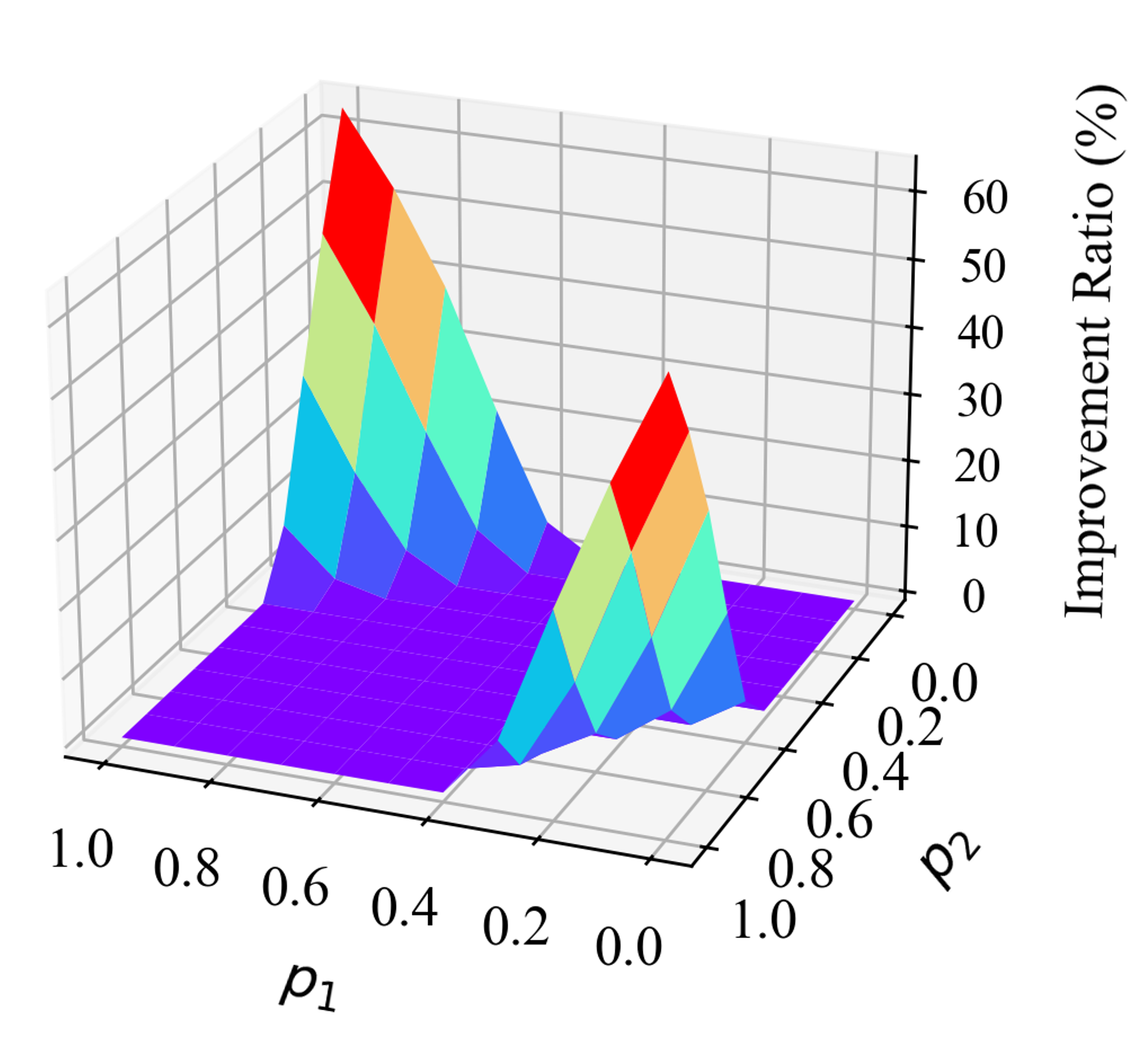}
    \caption{Impact of latency disparity when $w=0.2$.}
    \label{fig:example:disparity}
\end{figure}

This simple example demonstrates the substantial performance improvement achieved by optimally selecting latency compensation, particularly when synchronization performance is critical (e.g., with a small weighting parameter) or significant latency disparities across users. However, this requires knowledge of the latency distribution for all users, which is typically unavailable in practice. In the next section, we will develop a stochastic optimization algorithm that only requires latency samples from each user for each content delivery and exhibits fast convergence.

\section{Synchronized Scheduling Algorithm Design}
\label{sec:algorithm}

In this section, we develop a real-time synchronized scheduling algorithm for collaborative content delivery, which efficiently solves the stochastic optimization problem \eqref{eqn:obj} using latency measurement samples from each user for each content delivery. The proposed algorithm integrates the classic stochastic proximal point (SPP) method (see~\cite{rockafellar1976monotone,asi2019stochastic})  and the alternating direction method of multipliers (ADMM) (see \cite[Ch. 15]{beck2017first}), and the inherent problem structure to accelerate computation. Next, we discuss the motivation for employing SPP and ADMM methods to solve problem \eqref{eqn:obj}. 

Note that our objective function includes a non-smooth component resulting from the difference between the maximum and minimum latencies after latency compensation (see \eqref{eqn:example} for an example in the case with two users). As such, it is natural to employ the stochastic subgradient method. 
However, it is known that even in the deterministic setting, subgradient methods can suffer from numerical instabilities such as being sensitive to the step size choice and not being a descent method (i.e., the objective can increase between iterations). On the contrary, proximal point methods are generally more stable than gradient iterations and, in the deterministic setting, converge for a wide range of step size choices~\cite{bertsekas2011incremental}. Motivated by these numerical advantages, we apply its stochastic version as the framework for our algorithm design.

Recall that $\mb{d}_k$  is the random latency experienced by the users when the server delivers the $k^{th}$ content. Starting with some initial value $\mb{x}_0 \succeq \mb{0}$ for the latency compensation at $k =0$,  
we construct the SPP update rule:
\begin{align}\label{p:spp-subprob}
    \mb{x}_{k+1} = \argmin_{\mb{x} \succeq \mb{0}}\bigg\{f \big(\mb{x}; \mb{d}_k \big) + \frac{1}{2 \gamma_k} \| \mb{x} - \mb{x}_k\|^2 \bigg\},
\end{align}
where the step size $\gamma_k > 0$ will be specified in Proposition \ref{Proposition convergence}. 

The subproblem~\eqref{p:spp-subprob} is deterministic, strongly convex, and has a unique solution. A straightforward solution is to directly solve it using a general-purposed solver such as SDPT3~\cite{toh1999sdpt3} and SeDuMi. However, as the complexity of these solvers generally scales polynomially with the problem dimension, the latency generated by solving the problem can easily exceed $6$ ms for a moderate user number $N$, rendering the compensation control meaningless. As such, there is a pressing need for a low computational complexity algorithm for~\eqref{p:spp-subprob}, which will be developed next.

As the first step, we apply a change of variable and define $\mb{y} \triangleq \mb{x} + \mb{d}_k$ and $\mb{y}_k \triangleq \mb{x}_k + \mb{d}_k$. Substituting the expression of $f$ into~\eqref{p:spp-subprob}, we can equivalently rewrite it as
\begin{multline}\label{p:spp-subprob-equiv}
    \mb{y}_{k+1} =
      \argmin_{\mb{y} \succeq \mb{d}_k} \bigg\{ \max_{n \in [N]} y(n) - \min_{n \in [N]} y(n) \\
      + w \sum_{n=1}^{N } c_n(y(n))  + \frac{1}{2 \gamma_k} \| \mb{y} - \mb{y}_k\|^2\bigg\}.
\end{multline}

Denote $\widetilde{f}(\mb{y}) \triangleq \max_{n \in [N]} y(n) - \min_{n \in [N]} y(n)$.
The difficulty in solving~\eqref{p:spp-subprob-equiv} is the co-existence of the nonsmooth term $\widetilde{f}$ and the element-wise inequality constraint $\mb{y} \succeq \mb{d}_k$. This structure motivates applying the ADMM algorithm to simplify the computation by decoupling the two terms and just optimizing with respect to (w.r.t.) one of them per update.

\subsection{ADMM for Solving~\eqref{p:spp-subprob-equiv}}

Introducing auxiliary vector $\mb{z}$, we rewrite~\eqref{p:spp-subprob-equiv} as 
\begin{align} 
\begin{split}
     \min_{\mb{y}, \mathbf{z}} &  \left\{ \widetilde{f}(\mb{z}) + w \sum_{n=1}^{N } c_n(y(n)) + \frac{1}{2 \gamma_k} \| \mb{y} - \mb{y}_k\|^2\right\} \label{eqn:SPP:separate:1}\\
    \text{s.t.} & \quad \mathbf{z} = \mb{y} ,  \quad \mb{y} \succeq \mathbf{d}_k.
\end{split}
\end{align}
Let $\boldsymbol{\lambda} \in \mathbb{R}^N$ be the Lagrange multiplier vector corresponding to the equality constraint. After relaxing the equality, the augmented Lagrangian is given by:
\begin{align}
\begin{split}
          L_{\rho} (\mb{y}, \mb{z}, \boldsymbol{\lambda}) \triangleq \, &  \widetilde{f}(\mb{z}) +w \sum_{n=1}^{N } c_n(y(n)) + \frac{1}{2 \gamma_k} \| \mb{y} - \mb{y}_k\|^2  \\ 
    & + \langle \boldsymbol{\lambda}, \mb{z} - \mb{y}\rangle  + \frac{\rho}{2} \|\mb{z}- \mb{y} \|^2, 
\end{split}
\end{align}
where $\rho > 0 $ is a suitably chosen step size.

Let superscript $t \in \mathbb{Z}_+$ denote the  ADMM iteration index. Given  initialization $\boldsymbol{\lambda}_k^0 \in \mathbb{R}^N$, $\mb{z}_k^0 = \mb{y}_k^0 \succeq \mb{d}_k$, the variables are updated as:
\begin{align} 
 \mb{y}_k^{t+1}  = &\argmin_{\mb{y}\succeq \mathbf{d}_k}  \Big\{w \sum_{n=1}^{N } c_n'(y_k^t(n))y(n) + \frac{\eta}{2}\|\mb{y} - \mb{y}_k^{t} \|^2 \nonumber \\ 
 &\;+ \frac{1}{2\gamma_k} \|\mb{y} - \mb{y}_k\|^2  +\frac{\rho}{2} \Big\|\mathbf{z}_k^{t}- \mb{y} + \frac{1}{\rho} \boldsymbol{\lambda}_k^t \Big\|^2 
        \Big\}, \label{eqn:ADMM:x} \\
\mb{z}_k^{t+1}  = &\argmin_{\mb{z}} \widetilde{f}(\mb{z}) + \frac{\rho}{2} \Big\| \mb{z} - \mb{y}_k^{t+1} + \frac{1}{\rho} \boldsymbol{\lambda}_k^{t}\Big\|^2, \label{eqn:ADMM:z}  \\
\boldsymbol{\lambda}_k^{t+1}  =  &\boldsymbol{\lambda}_k^{t} + \rho \,(\mathbf{z}_k^{t+1} - \mb{y}_k^{t+1}), \label{eqn:ADMM:lambda}
\end{align}
where $c_n'(\cdot)$ representing the derivative of the cost function and $\eta$ is a tunable parameter. The $\mb{y}$-update \eqref{eqn:ADMM:x} follows a linearized ADMM scheme \cite{ouyang2015accelerated}, where the cost term is approximated at the current iterate using the first-order gradient information. Since $\mb{y}$-update \eqref{eqn:ADMM:x} is a separable convex optimization problem, each element admits a closed-form update:
\begin{align}
\hat{y}_k^{t+1}(n) &= \frac{ \frac{ y_k(n)}{\gamma_k} + \rho z_k^t(n) + \lambda_k^t(n) - w c_n'(y_k^t(n)) + \eta y_k^t(n)}{ \frac{1}{\gamma_k} + \rho  + \eta}, \nonumber \\
y_k^{t+1}(n) &= \max \{ d_k(n), \hat{y}_k^{t+1}(n) \}, \quad \forall n \in [N]. \label{eqn:ADMM:y_solution}
\end{align}

In particular, when $c_n(\cdot)$ is linear, setting $\eta = 0$ recovers the standard ADMM update, as shown in our prior work \cite{wang2025optimal}.  

It remains to specify how to compute the $\mb{z}$-update~\eqref{eqn:ADMM:z}. As such, we introduce the \emph{proximal mapping} 
\begin{align}\label{eq:prox}
\prox_{\widetilde{f}/\rho}(\mb{x}) \triangleq \argmin_{\mb{y}}  \widetilde{f}(\mb{y}) + \frac{\rho}{2} \| \mb{y} - \mb{x}\|^2.    
\end{align}
As such, the $\mb{z}$-update~\eqref{eqn:ADMM:z} corresponds to the calculation of $\prox_{\widetilde{f}/\rho}(\mb{y}^{t+1}-\frac{1}{\rho}\boldsymbol{\lambda}^{t})$. Next, we introduce a fast algorithm to compute the proximal mapping.

\subsection{Fast Computation of the Proximal Mapping $\prox_{\widetilde{f}/\rho}(\mb{x})$}

\wang{We note that $\widetilde{f}(\mb{y})$ is convex and $\frac{\rho}{2} \| \mb{y} - \mb{x}\|^2$ is $\rho$-strongly convex for $\rho >0$. Hence, the proximal subproblem is strongly convex, and $\prox_{\widetilde{f}/\rho}(\mb{x})$ is the unique optimal solution.}
Before introducing the algorithm for~\eqref{eq:prox}, we need the following lemma regarding the property of $\prox_{\widetilde{f}/\rho}$. For brevity, we denote $\prox_{\widetilde{f}/\rho}(\mb{x})$ as $\mb{y}^\star$.

\begin{lemma}\label{lem:order-preserve}
    Let $\mb{x}\triangleq(x(n))_{n=1}^{N}$ be sorted in non-decreasing order: $x(1) \leq x(2) \leq \cdots \leq x(N)$. Then the elements of $\prox_{\widetilde{f}/\rho}(\mb{x})$ must also be non-decreasing. 
\end{lemma}
\begin{proof}
    See Appendix~\ref{pf:lem-1}.
\end{proof}

Assisted by Lemma~\ref{lem:order-preserve}, we are now ready to introduce the algorithm for efficiently computing~\eqref{eq:prox}. First, note that problem~\eqref{eq:prox} is permutation invariant. Therefore, without loss of generality, we assume the elements of $\mb{x}$ are in non-decreasing order. According to Lemma~\ref{lem:order-preserve}, it must hold that $\argmax_{n \in [N]} y^\star(n) = \{N, N-1, \ldots, n_{\max}\}$ for some positive integer $n_{\max}\leq N$. That is, the indices for which $y^\star(n)$ attain its maximum must include $N$ and decrease consecutively until $n_{\max}$. Similarly, for the indices where $y^\star(n)$ attain its minimum, we have $\argmin_{n \in [N]} y^\star(n) = \{1, 2, \ldots, n_{\min}\}$ for some positive constant $n_{\min}$, where $n_{\min}\leq n_{\max}$. When $n_{\min}=n_{\max}$, then all elements of $\mb{y}^\star$ are equal. When $n_{\min}<n_{\max}$, then the elements of $\mb{y}^\star$ satisfies
\begin{multline}\label{eq:index-constr}
        y^\star(N)  = \cdots = y^\star(n_{\max}) > y^\star(n_{\max}-1) \\
        \geq \cdots  \geq y^\star(n_{\min} + 1) > y^\star(n_{\min}) = \cdots = y^\star(1).
\end{multline}

Given these nice structure properties, we can effectively calculate $\mb{y}^\star$ as follows.
\begin{proposition}\label{prop:prox-sol}
    Assume $\mb{x}$ has non-decreasing elements and let
    $n_{\min}$ and $n_{\max}$ be the indices satisfying~\eqref{eq:index-constr}. Then
    \begin{itemize}
        \item[(1)] When $n_{\min}=n_{\max}$, then 
        \begin{align}
            \label{prop:solution:case1}
            y^\star(1)=y^\star(2)=\cdots=y^\star(N)=\frac{1}{N}\sum_{n=1}^{N}x(n).
        \end{align}
        \item[(2)] When $n_{\min}< n_{\max}$, then 
    \end{itemize}
    \begin{align}\label{eq:prox-sol}
        y^\star(n) = 
        \begin{cases}
            x(n), &  n_{\min} < n < n_{\max}\\[1ex]
            \frac{\sum_{n = n_{\max}}^N x(n) -1/\rho}{N - n_{\max} + 1} , & n \geq n_{\max}\\[1ex]
             \frac{\sum_{n = 1}^{n_{\min}} x(n) + 1/\rho}{n_{\min}}, & n \leq n_{\min}.
        \end{cases}
    \end{align}
\end{proposition}

\begin{proof}
    In the setting of Prop.~\ref{prop:prox-sol}, by Lemma~\ref{lem:order-preserve}, when $n_{\min}=n_{\max}$, we have $y^\star(1)=y^\star(2)=\cdots=y^\star(N)=y$ and thus problem~\eqref{eq:prox} is equivalent to 
    \begin{align*}
        \min_{y} \frac{1}{2}\sum_{n=1}^{N}\left(y-x(n)\right)^2.
    \end{align*}
    It is clear that the optimal $y$ value is equal to the average of the elements in the vector $\mb{x}$.

    When $n_{\min}<n_{\max}$, problem~\eqref{eq:prox} is equivalent to 
    \begin{equation}
        \begin{aligned}
          & \ \underset{y(1), \ldots, y(N)}{\text{minimize}}
           && y(N) - y(1) + \frac{\rho}{2} \sum_{n = 1}^{N} (y(n) - x(n))^2\\
          & \text{subject to} && y(1) = \cdots = y(n_{\min}),\\
          &  \quad && y(n_{\max}) = \cdots = y(N).
        \end{aligned}
    \end{equation}
    Solving for each $y(n)$ element-wise completes the proof.
\end{proof}
  
\begin{algorithm}[!t]
\caption{Compute $\prox_{\widetilde{f}/\rho}(\mb{x})$}
\label{alg:prox-subroutine}
\begin{algorithmic}[1]
\Require Latency compensation vector $\mb{x}$
\Ensure The proximal mapping $\prox_{\widetilde{f}/\rho}(\mb{x})$, denoted as $\mb{y}^\star$
\State Sort $\mb{x}$ in non-decreasing order and keep track of the original indices
\State Initialize $n_{\min} \gets 1$ and $n_{\max} \gets N$
\While{true}
    \LineComment{Compute $y^\star(n)$ using Proposition \ref{prop:prox-sol}}
    \If{$n_{\max} = n_{\min}$}
        \State Using \eqref{prop:solution:case1} to compute $y^\star(n)$
        \State \textbf{break}
    \EndIf
    \If {\eqref{eq:index-constr} is satisfied}
        \State \textbf{break}
    \Else
        \If{$y^\star(n_{\max} - 1) \geq y^\star(n_{\max})$}
            \State $n_{\max} \gets n_{\max} - 1$
        \ElsIf{$y^\star(n_{\min} + 1) \leq y^\star(n_{\min})$}
            \State $n_{\min} \gets n_{\min} + 1$
        \EndIf
    \State Using \eqref{eq:prox-sol} to compute $y^\star(n)$
    \EndIf
\EndWhile
\State Resort $\mb{y}^\star$ to match the original indices of $\mb{x}$
\State \Return $\text{prox}_{\widetilde{f}/\rho}(\mb{x})$
\end{algorithmic}
\end{algorithm}

Proposition \ref{prop:prox-sol} reveals that computing $\prox_{\widetilde{f}/\rho}(\mb{x})$ is roughly averaging respectively the $n_{\min}$ smallest and $N - n_{\max} + 1$ largest elements of
$\mb{x}$, while leaving middle elements unchanged. This solution makes intuitively sense, as the objective of~\eqref{eq:prox} requires minimizing the max-min gap of $\mb{y}^\star$ while keeping it as close to $\mb{x}$ as possible. 
It remains to find $n_{\min}$ and $n_{\max}$ so that~\eqref{eq:index-constr} holds. 
The structure of~\eqref{eq:prox-sol} suggests we can search for $n_{\min}$ and $n_{\max}$ starting from the two ends towards the center. Specifically, we initialize the search with $n_{\min} = 1$ and $n_{\max} = N$. Then for each pair $(n_{\min}, n_{\max})$, we compute~\eqref{eq:prox-sol} and check if~\eqref{eq:index-constr} is satisfied. If the edge condition 
$y^\star(n_{\max}) > y^\star(n_{\max} - 1)$ (respectively $y^\star(n_{\min} + 1) > y^\star(n_{\min})$) is violated, we decrease $n_{\max}$
(respectively increase $n_{\min}$)  by $1$. Otherwise, the algorithm terminates.

For example, consider $\mb{x}=(1,1,2)$ and $\rho = 0.5$. To calculate the $\prox_{\widetilde{f}/\rho}(\mb{x})$, we first start $n_{\min} = 1$ and $n_{\max} = 3$. Following Proposition \ref{prop:prox-sol}, we have $\mb{y}^\star=(3,1,0)$ and thus $y^\star(n_{\max})<y^\star(n_{\max}-1)$ and $y^\star(n_{\min})>y^\star(n_{\min}+1)$. Hence, we have to increase $n_{\min}$ and decrease $n_{\max}$ by 1, respectively, which leads to $n_{\min} = n_{\max}=2$. Thus, according to Proposition \ref{prop:prox-sol}, we have $y^{\star} =(4/3,4/3,4/3)$. In another example, consider $\mb{x}=(1,2,5)$ and $\rho = 1$. To calculate the $\prox_{\widetilde{f}/\rho}(\mb{x})$, we first start $n_{\min} = 1$ and $n_{\max} = 3$. Following Algorithm 1, we have $\mb{y}^\star=(2,2,4)$, satisfying that $y^\star(1)\leq y^\star(2)\leq y^\star(3)$.

The procedure is summarized in Algorithm~\ref{alg:prox-subroutine}, with the worst-case complexity involving only $N$ scalar additions, divisions, and comparisons. This, combined with the computational complexity of sorting $N$ numbers, results in a total computational complexity of $\mathcal{O}(N\log N)$. 

Combining the above components, the Stochastic Content Scheduling (SCS) algorithm for the server for each content delivery is summarized in Algorithm~\ref{alg:SCS}. The SCS algorithm begins with an initial latency sample $\mathbf{d}_0$ for each user. For each $(k+1)^{th}$ content delivery, the latency compensation is determined by iteratively updating equations \eqref{eqn:ADMM:lambda}, \eqref{eqn:ADMM:y_solution}, and Algorithm $\ref{alg:prox-subroutine}$ over $T_k$ iterations. The content is then transmitted to user $n$ after a delay of $x_{k+1}(n)$ time units, and the latency during this delivery is collected to determine the latency compensation in the next delivery. \wang{This delivery-feedback loop repeats for each content, progressively refining the compensation vector to improve synchronization over time.}

\begin{algorithm}[htbp]
\caption{Synchronized Content Scheduling (SCS)}
\label{alg:SCS}
\begin{algorithmic}[1]
\Require Initial latency compensation vector $\mathbf{x}_0$, initial latency $\mathbf{d}_0$, number of contents $K$, and parameters $\rho$, $\eta$.
\State Initialize $\mathbf{x} \gets \mathbf{x}_0$
    \For{$k = 0$ to $K-1$}
        \State Initialize $\mathbf{y}_k^{0}$, $\mathbf{z}_k^{0}$, and $\boldsymbol{\lambda}_k^{0}$. 
        \For{$t = 0$ to $T_k-1$}
            \State Update $\mathbf{y}_k^{t+1}$ using \eqref{eqn:ADMM:y_solution}
            \State Update $\mathbf{z}_k^{t+1}$ using Algorithm \ref{alg:prox-subroutine} for computing $\prox_{\widetilde{f}/\rho}\left(\mb{y}_k^{t+1} - \frac{1}{\rho} \boldsymbol{\lambda}_k^{t}\right)$
            \State Update $\boldsymbol{\lambda}_k^{t+1}$ using \eqref{eqn:ADMM:lambda}
        \EndFor
        \State Compute the latency compensation $\mb{x}_{k+1} \gets \mb{y}^{T_{k}}_k - \mb d_k$
        \State Transmit the $(k+1)^{th}$ content to user $n$ after a delay of $x_{k+1}(n)$ time units, and collect the latency vector $\mathbf{d}_{k+1}$ during $(k+1)^{th}$ delivery
    \EndFor
\end{algorithmic}
\end{algorithm}

\section{Convergence Analysis}
\label{sec:convergence}

This section provides the convergence analysis of Algorithm~\ref{alg:SCS}. For convenience,  define $g(\mb{x};\mb{d}_{k}) \triangleq f(\mb{x};\mb{d}_{k}) + \mathbb{I}_{\mathbb{R}_+} (\mb{x})$ with $\mathbb{I}_{\mathbb{R}_+} $ being the indicator function of the nonnegative orthant. Specifically, $\mathbb{I}_{\mathbb{R}_+}(\mb{x})$ equals to $0$ if $\mb{x}\in\mathbb{R}_+$ and $+\infty$ if $\mb{x}\notin\mathbb{R}_+$. 
Since ADMM is an iterative algorithm that converges asymptotically, in practice, one cannot solve exactly the subproblem~\eqref{p:spp-subprob}. Therefore, we consider solving it inexactly by finding  an $\epsilon_k$-minimizer, i.e., $\mb{x}_{k+1}$  satisfies
\begin{align} 
\label{spp-subprob:inexact}
&g(\mb{x}_{k+1};\mb{d}_{k}) + \frac{1}{2\gamma_{k}} \left\|\mathbf{x}_{k+1} - \mathbf{x}_{k}\right\|^2 \nonumber \\ 
& \leq \underset{\mathbf{x}}{\min} \Bigg\{g(\mb{x};\mb{d}_{k}) + \frac{1}{2\gamma_{k}} \left\|\mathbf{x} - \mathbf{x}_{k}\right\|^2 \Bigg\} + \epsilon_{k}. 
\end{align}

Let $G(\mb{x}) \triangleq F(\mb{x}) + \mathbb{I}_{\mathbb{R}_+} (\mb{x})$. We note that the optimal solution to $\min_{\mb{x}\succeq\mb{0}}F(\mb{x})$ is equivalent to the optimal solution to $\min_{\mb{x}}G(\mb{x})$. Next, we show that our proposed SCS algorithm converges under suitably chosen $\epsilon_k$ in each iteration $k$, as stated in the following proposition.

\begin{proposition}
\label{Proposition convergence}
Let \(\{\mb{x}_k\}_{k\geq 0}\) be the sequence generated by Algorithm~\ref{alg:SCS}, and let \(\mb{x}^\star\) be a minimizer of \(G\). Then, for all $K\geq 1$, it holds  
    \begin{align*}
        \mathbb{E}\left[G(\overline{\mb{x}}_K)\right] - G(\mb{x}^\star)  \quad \leq \frac{\left(\|\mb{x}_0 - \mb{x}^\star\| + 2A_K + \sqrt{2B_K}\right)^2}{2 \sum_{k=0}^{K} \gamma_k},
    \end{align*}
    where $\overline{\mb{x}}_K \triangleq \frac{1}{\sum_{k=0}^{K} \gamma_k} \sum_{k=0}^{K} \gamma_k \mb{x}_k$, $A_K \triangleq \sum_{k=0}^{K} \sqrt{2\gamma_{k}\epsilon_{k}}$, and $B_K \triangleq \sum_{k=0}^{K}\left(\gamma_{k}\epsilon_{k} + \gamma_{k}^2 \mathbb{E}[\|f'(\mb{x}^\star;\mb{d}_k)\|^2]/2\right)$.

    Moreover, suppose the stepsizes satisfy $\gamma_k = \gamma_0 / \sqrt{K}$ for some $\gamma_0 > 0$, and the $\epsilon_k$ is decreased as $\mathcal{O}(1/k^p)$  with $p > 2$. Then we see that 
    \[
    \mathbb{E}[G(\overline{\mb{x}}_K)] - G(\mb{x}^\star) \leq \mathcal{O}\left(\frac{1}{\sqrt{K}}\right).
    \]
\end{proposition}

\begin{proof}
We start from the optimality condition of the subproblem~\eqref{spp-subprob:inexact}. 
   By~\cite[Lemma 2]{schmidt2011convergence}, for any $\mb{x}_{k+1}$ satisfying~\eqref{spp-subprob:inexact} there exists some vector $\mb{e}_k$ such that $\|\mb{e}_k\| \leq \sqrt{2 \gamma_k \epsilon_k}$ and 
    \begin{align}
         - \frac{1}{\gamma_k}(\mb{x}_{k+1} - \mb{x}_{k} + \mb{e}_k) \in \partial_{\epsilon_k} g(\mb{x}_{k+1}; \mb{d}_k),
    \end{align}
    where $\partial_{\epsilon}g(\mb{x};\mb{d})$ is the $\epsilon$-subdifferential set of $g( \cdot; \mb{d})$ at $\mb{x}$ defined as 
    \begin{equation*}
        \partial_{\epsilon}g(\mb{x};\mb{d}) \triangleq \{ \boldsymbol{\xi} \,|\, g(\mb{y};\mb{d}) \geq g(\mb{x};\mb{d}) + \boldsymbol{\xi}^T (\mb{y}  - \mb{x}) - \epsilon, \forall \mb{y}\}.
    \end{equation*}

Letting $\mb{y} = \mb{x}^\star$, $\mb{x} = \mb{x}_{k+1}$ and $\boldsymbol{\xi} =  - \frac{1}{\gamma_k} \left(\mb{x}_{k+1} - \mb{x}_{k} + \mb{e}_k \right)$ in $\partial_{\epsilon}g(\mb{x};\mb{d})$, thus gives 
    \begin{equation}
        \begin{split}
        &g(\mb{x}_{k+1};\mb{d}_k) \\
         \leq \ & g(\mb{x}^\star;\mb{d}_k) + \Big \langle \frac{\mb{x}_{k+1} - \mb{x}_{k} + \mb{e}_k }{\gamma_k}, \mb{x}^\star - \mb{x}_{k+1}  \Big \rangle + \epsilon_{k} \\   
         = \ & g(\mb{x}^\star;\mb{d}_k) + \frac{1}{2\gamma_k}\|\mb{x}_k-\mb{x}^\star\|^2 - \frac{1}{2\gamma_k}\|\mb{x}_{k+1}-\mb{x}^\star\|^2  \\
          &+ \Big \langle \mb{x}^\star-\mb{x}_{k+1},  \frac{\mb{e}_k}{\gamma_k} \Big \rangle -\frac{1}{2\gamma_k}\|\mb{x}_{k+1}-\mb{x}_k\|^2  + \epsilon_k\\
          \leq \ & g(\mb{x}^\star;\mb{d}_k) + \frac{1}{2\gamma_k}\|\mb{x}_k-\mb{x}^\star\|^2 - \frac{1}{2\gamma_k}\|\mb{x}_{k+1}-\mb{x}^\star\|^2 \\
          &+ \frac{1}{\gamma_k} \|\mb{e}_k\|\|\mb{x}^\star-\mb{x}_{k+1}\|  -\frac{1}{2\gamma_k}\|\mb{x}_{k+1}-\mb{x}_k\|^2 + \epsilon_k \\
          \leq \ & g(\mb{x}^\star;\mb{d}_k) + \frac{1}{2\gamma_k}\|\mb{x}_k-\mb{x}^\star\|^2 - \frac{1}{2\gamma_k}\|\mb{x}_{k+1}-\mb{x}^\star\|^2\\
          &+ \sqrt{\frac{2\epsilon_k}{\gamma_k}}\|\mb{x}^\star-\mb{x}_{k+1}\|  -\frac{1}{2\gamma_k}\|\mb{x}_{k+1}-\mb{x}_k\|^2 + \epsilon_k
        \end{split}
        \label{proof:bound:v1}
    \end{equation}
where the third step follows from Cauchy–Schwarz inequality.

Rearranging terms on \eqref{proof:bound:v1}, we can get for any subgradient $g'(\mb{x}^\star; \mb{d}_k) \in \partial g(\mb{x}^\star; \mb{d}_k)$\footnote{The \textit{subdifferential} of $f$ at $x$, denoted by $ \partial f(x) $, is defined as  $\partial f(x) \triangleq \left\{ \xi | f(y) \geq f(x) + \langle \xi, y - x \rangle, \forall y \in \text{dom}f \right\}$, and $\text{dom}f$ is the domain of $f$.
} it must be
\begin{align*}
    \begin{split}
         & \frac{1}{2\gamma_k}\|\mb{x}_{k+1}-\mb{x}^*\|^2 \\
          \leq &   \  \big(g(\mb{x}^\star; \mb{d}_k) - g(\mb{x}_{k+1}; \mb{d}_k)\big) + \frac{1}{2 \gamma_k}\|\mb{x}_{k}-\mb{x}^*\|^2\\
          & + \sqrt{\frac{2 \epsilon_k}{\gamma_k}} \|\mb{x}_{k+1} - \mb{x}^\star\| - \frac{1}{2\gamma_{k}} \|\mathbf{x}_{k+1} - \mathbf{x}_{k}\|^2 + \epsilon_k  \\
          \leq &  \ \langle g'(\mb{x}^\star; \mb{d}_k), \mb{x}^\star - \mb{x}_k  \rangle + \frac{1}{2 \beta} \|\mb{x}_k - \mb{x}_{k+1}\|^2   \\
            & + \frac{\beta}{2}\|g'(\mb{x}^\star;\mb{d}_k)\|^2 + \frac{1}{2 \gamma_k}\|\mb{x}_{k}-\mb{x}^*\|^2 + \epsilon_k    \\
            & + \sqrt{\frac{2 \epsilon_k}{\gamma_k}} \|\mb{x}_{k+1} - \mb{x}^\star\| - \frac{1}{2\gamma_{k}} \|\mathbf{x}_{k+1} - \mathbf{x}_{k}\|^2\\
          \leq &  \ - \big(g(\mb{x}_{k};\mb{d}_k) - g(\mb{x}^\star;\mb{d}_k)\big) + \frac{1}{2 \beta} \|\mb{x}_k - \mb{x}_{k+1}\|^2  \\
            & + \frac{1}{2 \gamma_k}\|\mb{x}_{k}-\mb{x}^*\|^2  + \sqrt{\frac{2 \epsilon_k}{\gamma_k}} \|\mb{x}_{k+1} - \mb{x}^\star\| + \epsilon_k  \\
            &- \frac{1}{2\gamma_{k}} \|\mathbf{x}_{k+1} - \mathbf{x}_{k}\|^2  + \frac{\beta}{2}\|g'(\mb{x}^\star;\mb{d}_k)\|^2 
    \end{split}
\end{align*}
where the second step follows \cite[Lemma 3.2]{asi2019stochastic}, $\beta > 0$ is an adjustable constant, and the last step is by the first-order conditions of convexity. Choosing $\beta = \gamma_k$ yields 
\begin{align}
    \label{eq:descent}
    &\frac{1}{2\gamma_k}\|\mb{x}_{k+1}-\mb{x}^\star\|^2 \nonumber\\
    \leq & - \big(g(\mb{x}_{k};\mb{d}_k,) - g(\mb{x}^\star;\mb{d}_k)\big) + \frac{1}{2\gamma_k}\|\mb{x}_{k}-\mb{x}^\star\|^2  \nonumber\\ 
    & +\frac{\gamma_k}{2} \|g'(\mb{x}^\star;\mb{d}_k)\|^2 + \sqrt{\frac{2\epsilon_k}{\gamma_k}} \|\mb{x}_{k+1} - \mb{x}^\star\| + \epsilon_{k}. 
\end{align}

Define filtration $\mathcal{F}_k \triangleq \sigma(\mb{d}_1,...,\mb{d}_k)$. It holds $\mb{x}_{k} \in \mathcal{F}_{k-1}$. 
Taking the conditional expectation of~\eqref{eq:descent} on both sides w.r.t. $\mathcal{F}_{k-1}$ we have
\begin{align*} \label{proof:conditonexpect}
    \begin{split}
    & \frac{1}{2\gamma_k}\mathbb{E}[\|\mb{x}_{k+1} - \mb{x}^\star\|^2 | \mathcal{F}_{k-1}] \\
    \leq &  - \big(G(\mb{x}_{k}) - G(\mb{x}^\star)\big)  +\frac{\gamma_k}{2} \mathbb{E}[\|g'(\mb{x}^\star; \mb{d}_k)\|^2 ]\\
    & + \sqrt{\frac{2\epsilon_k}{\gamma_k}} \mathbb{E}[\|\mb{x}_{k+1} - \mb{x}^\star\|^2 | \mathcal{F}_{k-1}] + \frac{1}{2\gamma_k} \|\mb{x}_{k}-\mb{x}^\star\|^2 + \epsilon_{k},
    \end{split}
\end{align*}
where we use the fact that $\mb{d}_k$ is independent of $\mathcal{F}_{k-1}$.

Rearranging terms,  multiplying $2\gamma_k$ to both sides of the inequality, and taking expectation w.r.t. $\mathcal{F}_{k-1}$, we obtain
\begin{align}
\begin{split}
        &2 \gamma_k \mathbb{E}[G(\mb{x}_{k})- G(\mb{x}^\star)] + \mathbb{E}[\|\mb{x}_{k+1}-\mb{x}^\star\|^2] 
    \\ \leq \ & \mathbb{E}[\|\mb{x}_{k}-\mb{x}^\star\|^2] + 2\sqrt{2\epsilon_{k}\gamma_k} \mathbb{E}[\|\mb{x}_{k+1}-\mb{x}^\star\|] 
    \\ &+ \gamma_k^2 \mathbb{E}[\|g'(\mb{x}^\star; \mb{d}_k)\|^2] + 2\gamma_k\epsilon_{k}.
\end{split}
\end{align}

Summing up over $k=0,...,K$, and using telescopic cancellation, we have the result
\begin{align}\label{eq:potential-function}
    &\sum_{k=0}^{K} 2\gamma_k \mathbb{E}[G(\mb{x}_k) - G(\mb{x^\star})] + \mathbb{E}[\|\mb{x}_{K+1} -\mb{x}^\star\|^2] \nonumber \\ 
    \leq& \ \mathbb{E}[\|\mb{x}_0 -\mb{x}^\star\|^2] + \sum_{k=0}^{K} {\gamma_k^2} \mathbb{E}[\|g'(\mb{x}^\star; \mb{d}_k)\|^2] \nonumber \\ 
    &+ \sum_{k=0}^{K}\left(2\sqrt{2\epsilon_{k}\gamma_k} \mathbb{E}[\|\mb{x}^\star-\mb{x}_{k+1}\|] + 2\gamma_k\epsilon_{k}\right).
\end{align}
Note that~\eqref{eq:potential-function} holds for any subgradient $g'(\mb{x}^\star; \mb{d}_k)$. Since $\partial g (\mb{x}^\star; \mb{d}_k) = \partial f (\mb{x}^\star; \mb{d}_k) + N_{\mathbb{R}_+} (\mb{x})$ with $N_{\mathbb{R}_+} (\mb{x}^\star)$ being the normal cone of $\mathbb{R}_+$ at $\mb{x}^\star$~\cite[Thm. 2.26]{royset2022optimization}, and $\mb{x}^\star$ must be nonnegative,  we can always pick $g'(\mb{x}^\star, \mb{d}_k)$ such that $\mathbb{E}[\|g'(\mb{x}^\star; \mb{d}_k)\|^2] = \mathbb{E}[\|f'(\mb{x}^\star; \mb{d}_k)\|^2] $, which is bounded by Lemma~\ref{lemma:subgrad}. 

The rest of the proof follows~\cite{schmidt2011convergence} with a slight modification. Denote $A_K \triangleq \sum_{k=0}^{K} \sqrt{2\gamma_{k}\epsilon_{k}}$ and $B_K \triangleq \sum_{k=0}^{K}\left(\gamma_{k}\epsilon_{k} + \gamma_{k}^2 \mathbb{E}[\|f'(\mb{x}^\star;\mb{d}_k)\|^2]/2\right)$. We will have 
\begin{align}
\sum_{k=0}^{K}\gamma_k\mathbb{E}\big[G(\mb{x}_k) - G(\mb{x}^\star)\big] \leq \frac{\left(\|\mb{x}_0 - \mb{x}^\star\| + 2A_K + \sqrt{2B_K}\right)^2}{2}  \nonumber.
\end{align}
As $G(\cdot)$ is convex, with the Jensen's inequality, we get,
\begin{align}
    \label{proof:convg:bound:v4}
    \mathbb{E}[G(\overline{\mb{x}}_K)] -  G(\mb{x}^\star) \leq \frac{(\|\mb{x}_0- \mb{x}^\star\| + 2A_K + \sqrt{2B_K})^2}{2\sum_{k=0}^{K} \gamma_{k}} \nonumber
\end{align}
where $\overline{\mb{x}}_K \triangleq \frac{1}{\sum_{k=0}^{K} \gamma_k}{\sum_{k=0}^{K}} \gamma_{k} \mb{x}_k$.

Set the stepsize  $\gamma_k \triangleq \gamma_0/\sqrt{K}$ for $k\geq 0$ and $\epsilon_k$ is decreased as $\mathcal{O}(1/k^p)$, the desired result follows from the following Lemma.

\begin{lemma}
\label{lemma:subgrad}
Let $\partial f(\mathbf{x}; \mathbf{d})$ denote the subdifferential of the function $f(\cdot\,; \mathbf{d})$ as defined in~\eqref{eqn:obj}, evaluated at $\mathbf{x}$. For some constant $L \geq 0$, and for any subgradient $f'(\mathbf{x}; \mathbf{d}) \in \partial f(\mathbf{x}; \mathbf{d})$, it holds 
    \begin{equation}
    \mathbb{E}\left[\left\| f'(\mathbf{x}; \mathbf{d}) \right\|^2\right]
    \leq 4 +  2Nw^2 L ^2.
\end{equation}
\end{lemma}
The proof is available in Appendix~\ref{Appendix:subdifferential}.
\end{proof}

Since linearized ADMM converges at a rate of $\mathcal{O}(1/t)$ for convex problems~\cite[Thm. 2.5]{ouyang2015accelerated},
the iteration number $T_k$ in Algorithm~\ref{alg:SCS} can be set as $\Omega(k^p)$ to ensure that $\epsilon_k=\mathcal{O}(1/k^p)$.

\section{Synchronization with Packet Drops}

\label{sec:dropping}
\begin{algorithm}[!t] 
\caption{Dropping-Case Variant of SCS (D-SCS)}
\label{alg:SCS_dropping}
\begin{algorithmic}[1]
\Require Initial latency vector $\mathbf{d}_0$, acknowledged set $\mathcal{H}_0$, zero compensation $\mathbf{x}_0 = \widetilde{\mathbf{x}}_0$, and maximum waiting time $\tau_{\max}$.
    \For{$k = 1$ to $K$}, 
        \State Set $\mathbf{x}_k \leftarrow \widetilde{\mathbf{x}}_k$
        \State Solve the subproblem over $\mathcal{H}_k$:
        \begin{equation}
            \mb{x}_{k+1} = \underset{\mb{x} \succeq \mb{0}}{\argmin} \left\{ f(\mb{x}; \mb{d}_k, \mathcal{H}_k) + \frac{1}{2\gamma_k} \|\mb{x} -\mb{x}_k\|_{\mc{H}_k}^2 \right\}
            \label{eqn:SPP:dropping}
        \end{equation}
        \State Determine the postpone time of the $(k+1)^{th}$ delivery $\widetilde{\mb{x}}_{k+1}\triangleq (\widetilde{x}_{k+1}(n))_{n=1}^N$:
        \begin{align*}
            \widetilde{x}_{k+1}(n) \gets
            \begin{cases}
            x_{k+1}(n), \quad &\text{if }  n \in \mathcal{H}_k \\
            \widetilde{x}_{k}(n), \quad &\text{otherwise } 
            \end{cases}
        \end{align*}
        \State Deliver the $(k+1)^{th}$ content to each user $n$ after a delay of $\widetilde{x}_{k+1}(n)$ time units.
        \State Wait for ACKs up to time limit $\tau_{\max}$, determine the acknowledged set $\mathcal{H}_{k+1}$, and observe $\mathbf{d}_{k+1}$ latency.
    \EndFor 
\end{algorithmic}
\end{algorithm}

We extend the framework to the dropping case, where only a subset of users may receive each delivery due to packet losses. In practice, a server expects an acknowledgment (ACK) from each user within a maximum tolerable latency $\tau_{\max}$ (e.g., $2$ or $3$ frames in an MR system). If no ACK is received in time, the user is considered to have missed the content. Let $I_k(n) \in \{0, 1\}$ indicate whether user $n$ successfully receives the $k^{th}$ content, with $I_k(n) = 1$ implying content $k$ having been successfully delivered to user $n$ and $I_k(n) = 0$ otherwise. We assume that the dropping event $I_k(n)$ is i.i.d. across delivery rounds $k$ and independent across users $n$. Define the corresponding acknowledged set as $$\mathcal{H}_k \triangleq \left\{n\in[N]: I_k(n)=1\right\},$$ which captures the users for whom content delivery succeeds in the $k^{th}$ delivery. Note that the acknowledged set $\mathcal{H}_k$ is i.i.d. across $k$, as each set is generated from the independent i.i.d. dropping event $I_k(n)$. In this setting, we aim to maintain synchronization within the acknowledged set $\mathcal{H}_k$ by minimizing the latency discrepancy while avoiding excessive latency cost. Since both the latency $\mb{d}_k$ and the acknowledged set $\mathcal{H}_k$ are i.i.d., we omit the subscript $k$ in the optimization formulation for notational simplicity. This leads to the following stochastic optimization problem,
\begin{equation}
    \operatorname*{\min}_{\mathbf{x} \succeq 0} F(\mathbf{x}) 
    \triangleq \mathbb{E} \left[ f(\mathbf{x}; \mathbf{d}, \mathcal{H}) \right],
    \label{system:obj:dropping}
\end{equation}
\wang{
where we define the synchronization metric as 
\begin{multline}
    f(\mathbf{x}; \mathbf{d}, \mathcal{H}) \triangleq \max_{n\in \mathcal{H}} (x(n) + d(n)) - \min_{n\in \mathcal{H}} (x(n) + d(n))  \\ + w \sum_{n\in \mathcal{H}} c_n(x(n) + d(n))
    \label{system:obj:dropping:instant}
\end{multline}
}
The dropping case~\eqref{system:obj:dropping} differs from the full-delivery case~\eqref{eqn:obj} in the randomness of the $\mc{H}$, making the objective dependent on both the random latency vector $d$ and the acknowledged set $\mc{H}$. Notably, when the network conditions are stable and the dropping probability satisfies $\mathbb{P}[I_k = 0] = 0$ for all $n\in[N]$, the formulation reduces to the original full-delivery model with $\mc{H} = [N]$. 

To accommodate dynamic network conditions with packet drops, we extend the SCS algorithm, named D-SCS, as detailed in Algorithm~\ref{alg:SCS_dropping}. \wang{Note that the waiting step is a feedback cutoff, not a blocking operation: the server continues rendering and transmitting contents using the latest compensation vector, while $\tau_{\max}$ only determines which ACKs are used to update the next compensation decision.} In each round, the proximal update is computed only for the acknowledged users in $\mathcal{H}_k$, leaving compensation values for $n \notin \mathcal{H}_k$ unchanged. \wang{Specifically, for users $n \notin \mathcal{H}_k
$, no ACK is received within $\tau_{\max}$, so their compensation values are carried over as  $\widetilde{x}_{k+1}(n) = \widetilde{x}_{k}(n)$, following step 4 of Algorithm~\ref{alg:SCS_dropping}.} Thus, the resulting subproblem is now restricted to the smaller subset $\mc{H}_k \subseteq [N]$. As demonstrated in later simulations (see \ref{sim:dropping_case}), this modification preserves the algorithm's ability to maintain synchronization under realistic drop rates. The empirical results strongly suggest that the convergence behavior of D-SCS is preserved. However, formal analysis remains challenging, as packet drops induce random, time-varying coordinate updates that break the telescoping structure used in the SCS convergence proof and obscure the influence of the dropping probability. A rigorous theoretical treatment of D-SCS is thus a promising direction for future work.

\section{Simulations}  
\label{sec:sim}

In this section, we first conduct synthetic simulations to validate the performance of our proposed synchronized scheduling algorithm. Then, we demonstrate the superiority of our algorithm using two real-world latency datasets: one from network latencies in Seattle \cite{zhu2016network} and the other from content delivery latency in MR-based firefighter training \cite{dong2023collaborative}. Additionally, we show that the proposed algorithm can handle the dropping case in realistic network conditions. \wang{Throughout, we use the improvement ratio (defined in~\eqref{eqn:improveratio}) as the primary performance metric, measuring the relative reduction in the objective function compared to the no-compensation case.}

\textbf{Hardware}. The simulations are conducted on a Dell Precision 7920 workstation equipped with Intel(R) Xeon(R) Gold $6242$R CPU @3.10Ghz, and $128$ GB RAM.

\textbf{Baseline}. We compare our proposed algorithm with the following three baselines.

\textit{SRA:} Simple running average (SRA) is a heuristic latency compensation method. It determines latency compensation by solving the optimization problem $\eqref{eqn:heuristic}$ using the running average of 
\wang{all previous latency samples.}

\textit{ERA:} \wang{Exponential running average (ERA) solves $\eqref{eqn:heuristic}$ using an exponentially weighted running average of latency samples, where the newest latency sample receives weight 0.1 and the previous estimate receives weight 0.9.}

\textit{VAR:} Variance-minimization (VAR) determines latency compensation by minimizing the compensated latencies across all users rather than the stricter max-min gap\wang{,  while keeping the same latency-cost term.} 

\textit{SSGD:} \wang{Stochastic subgradient descent (SSGD) is a first-order baseline that directly updates the compensation vector using a stochastic subgradient of $f(\mb{x};\mb{d}_k)$ at each delivery round, followed by projection onto $\mb{x}\succeq \mb{0}$.}

\wang{}

\begin{figure*}[!t]
    \centering
    \begin{subfigure}[b]{0.24\textwidth}
        \centering
        \includegraphics[width=\linewidth]{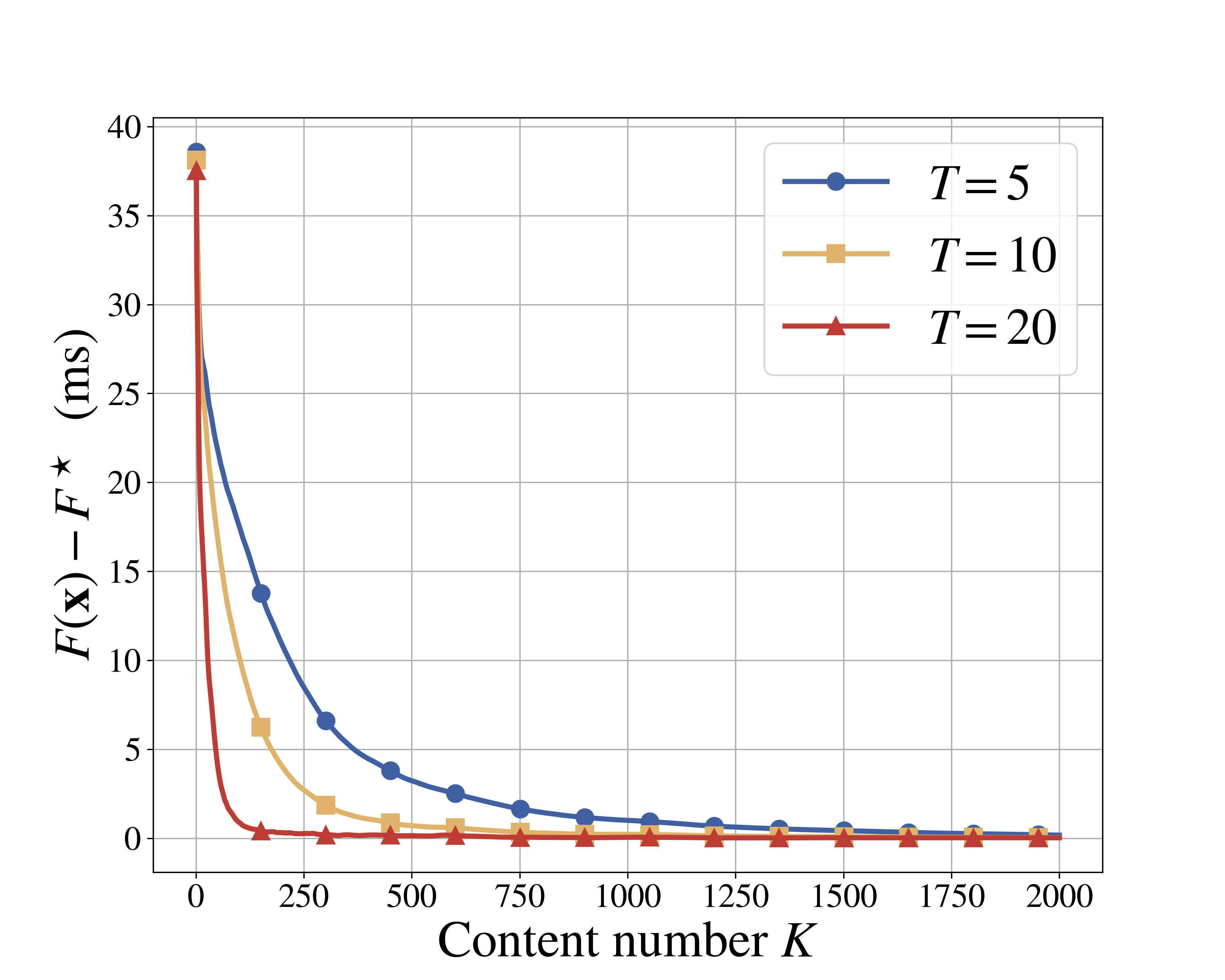}
        \caption{Truncated normal ($N=5$)}
        \label{fig:exp:syn:conv:TG:N5}
    \end{subfigure}
    \begin{subfigure}[b]{0.24\textwidth}
        \centering
        \includegraphics[width=\linewidth]{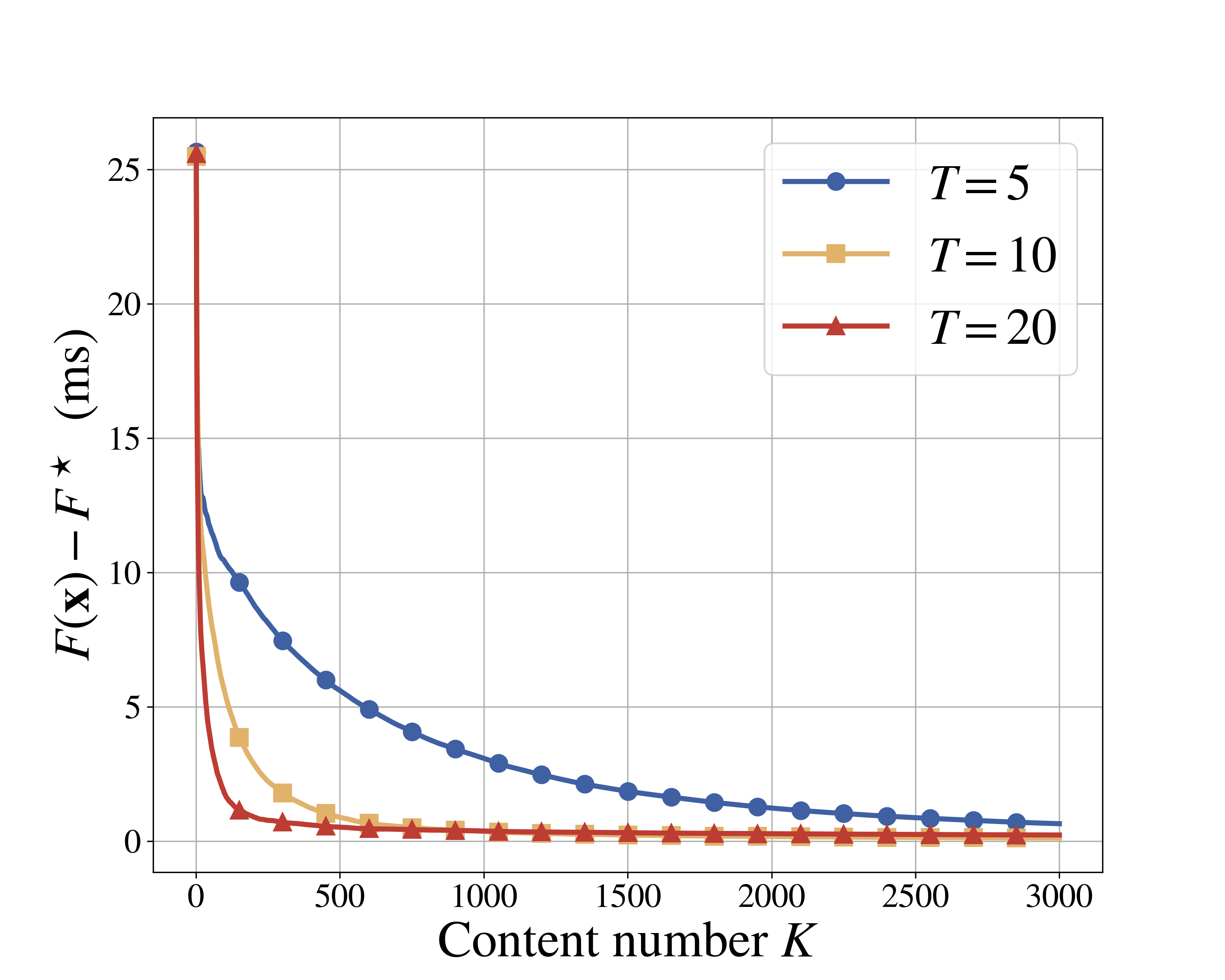}
        \caption{Truncated normal ($N=20$)}
        \label{fig:exp:syn:conv:TG:N20}
    \end{subfigure}
    \hfill
    \begin{subfigure}[b]{0.24\textwidth}
        \centering
        \includegraphics[width=\linewidth]{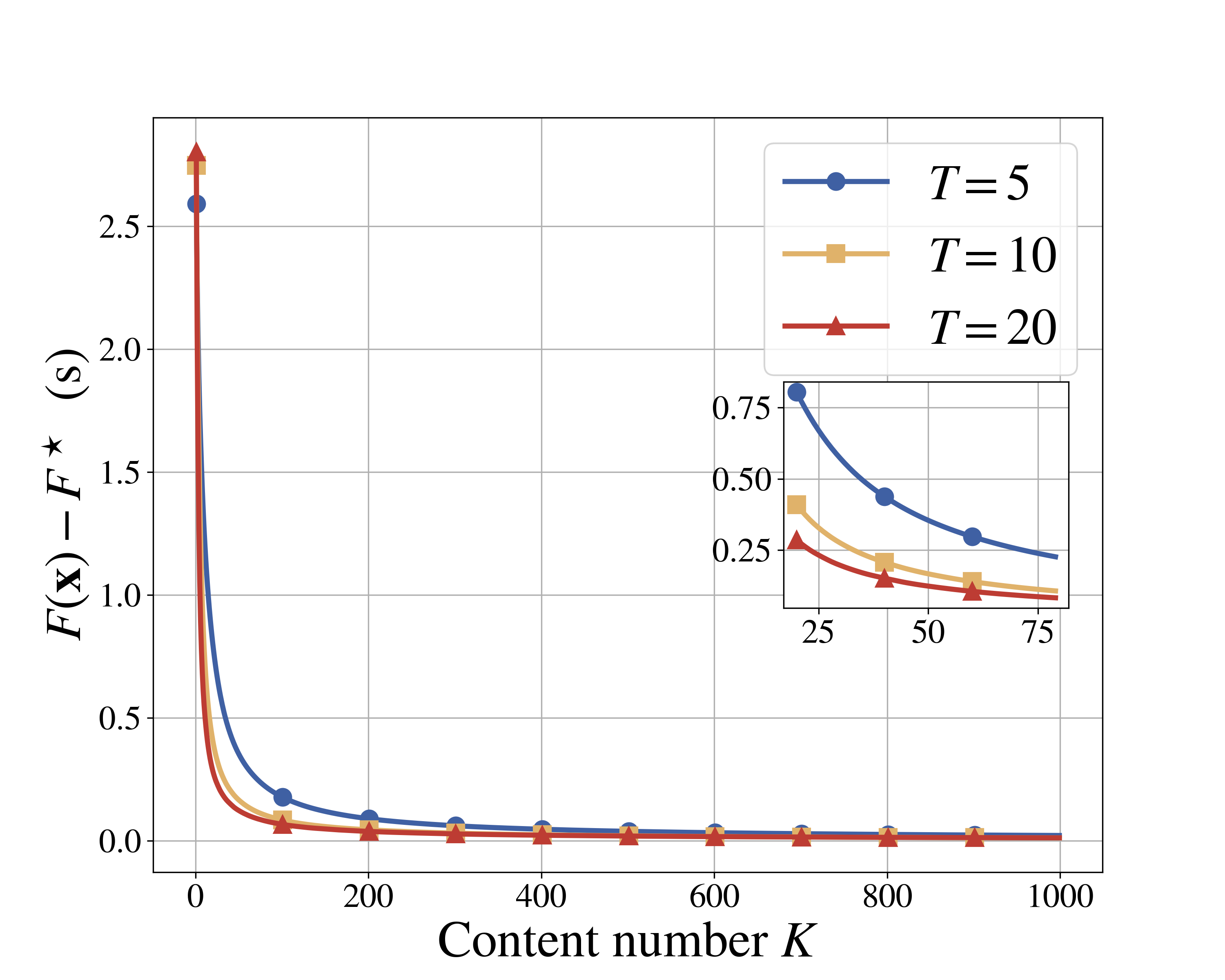}
        \caption{Log-normal ($N=50$)}
        \label{fig:exp:syn:conv:LN:N50}
    \end{subfigure}
    \hfill
    \begin{subfigure}[b]{0.24\textwidth}
        \centering
         \includegraphics[width=\linewidth]{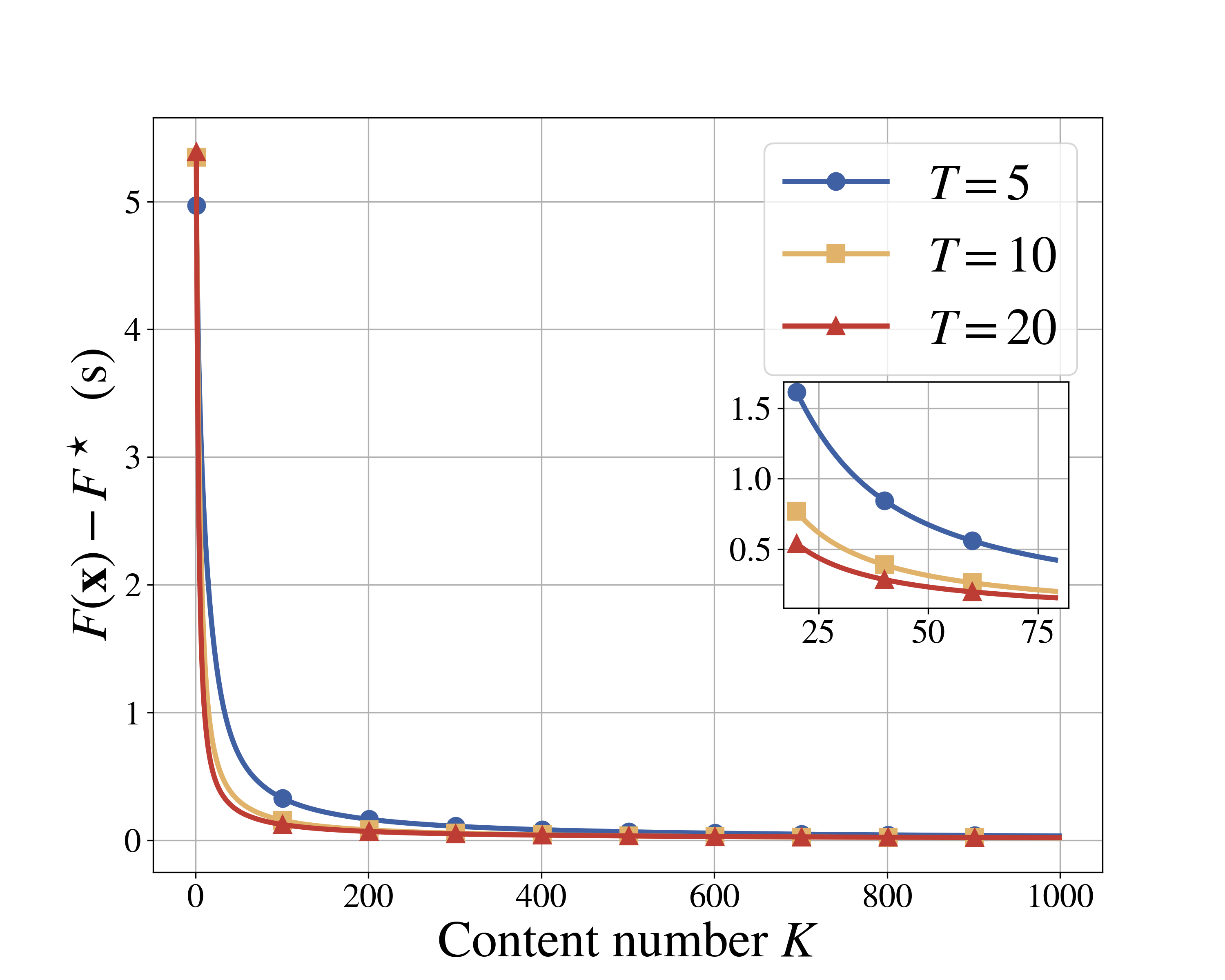}
        \caption{Log-normal ($N=100$)}
        \label{fig:exp:syn:conv:LN:N100}
    \end{subfigure}
    \hfill
    \caption{Synthetic simulation: Impact of parameter $T$.}
    \label{fig:exp:syn:convgence}
\end{figure*}

\begin{figure*}[!pt]  %
    \begin{subfigure}[b]{0.24\textwidth}
        \centering
        \includegraphics[width=\linewidth]{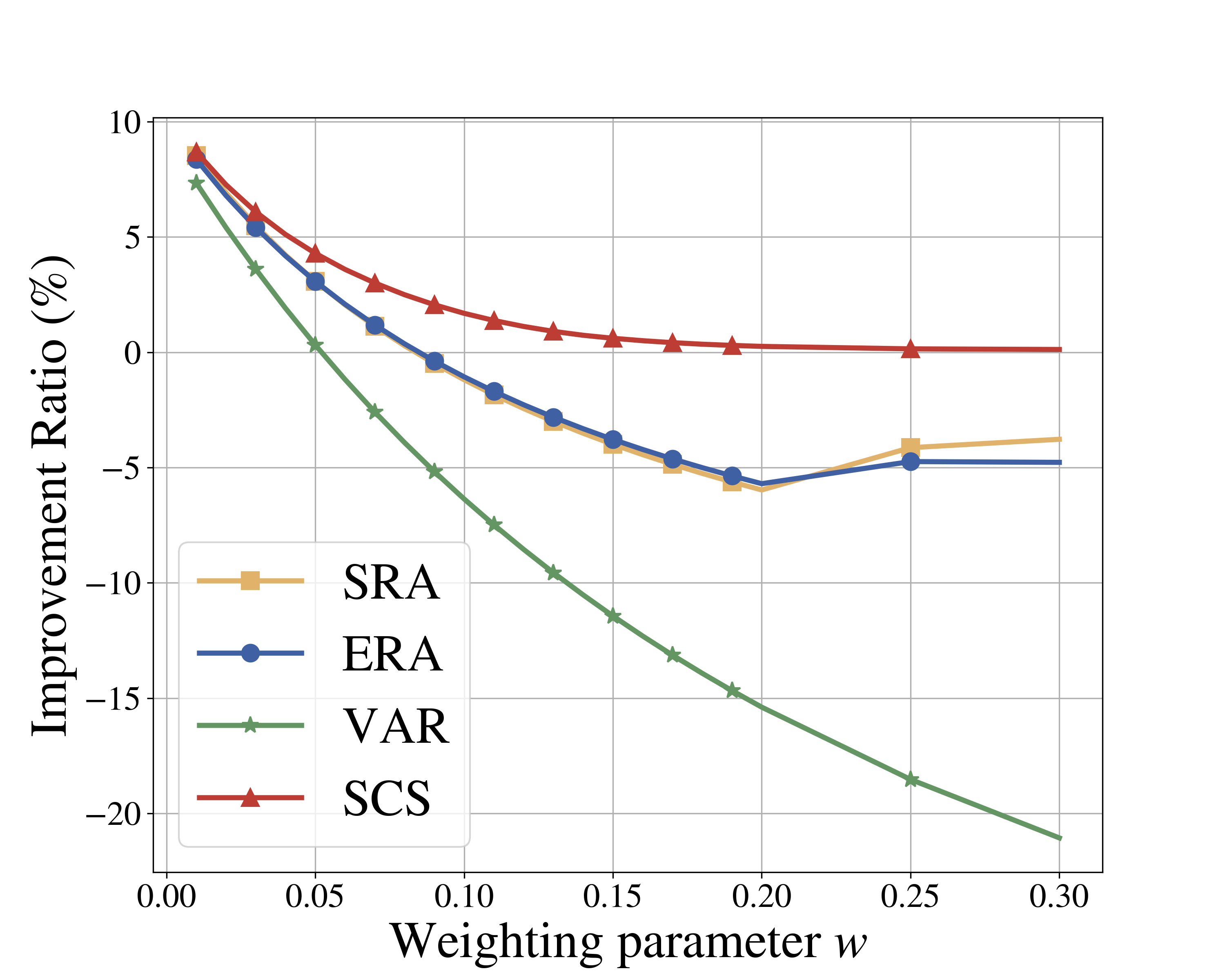}
        \caption{Truncated normal ($N=5$)}
        \label{fig:exp:syn:TG:N5}
    \end{subfigure}
    \hfill
    \begin{subfigure}[b]{0.24\textwidth}
        \centering
        \includegraphics[width=\linewidth]{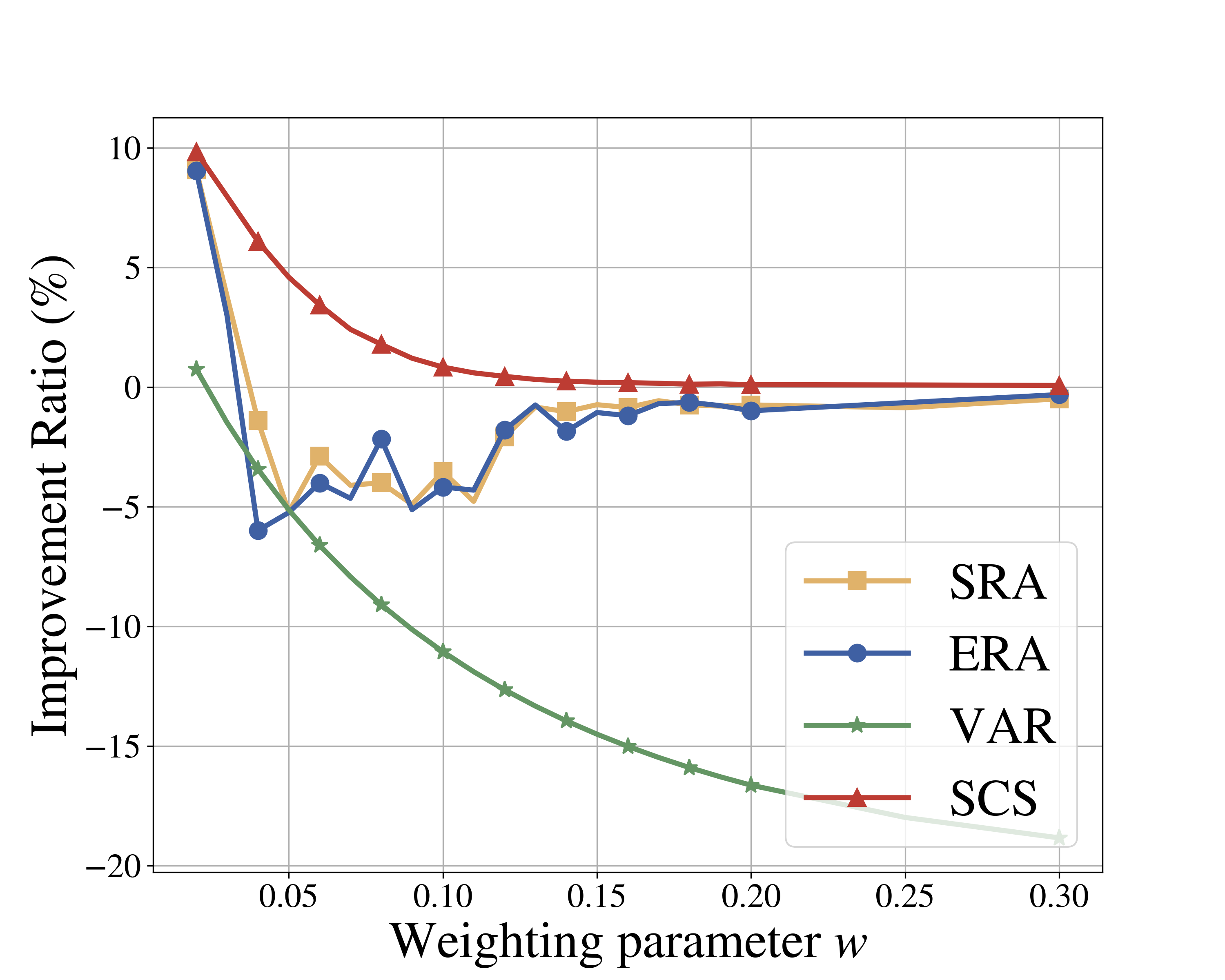}
        \caption{Truncated normal ($N=20$)}
        \label{fig:exp:syn:TG:N20}
    \end{subfigure}
    \hfill
    \begin{subfigure}[b]{0.24\textwidth}
        \centering
         \includegraphics[width=\linewidth]{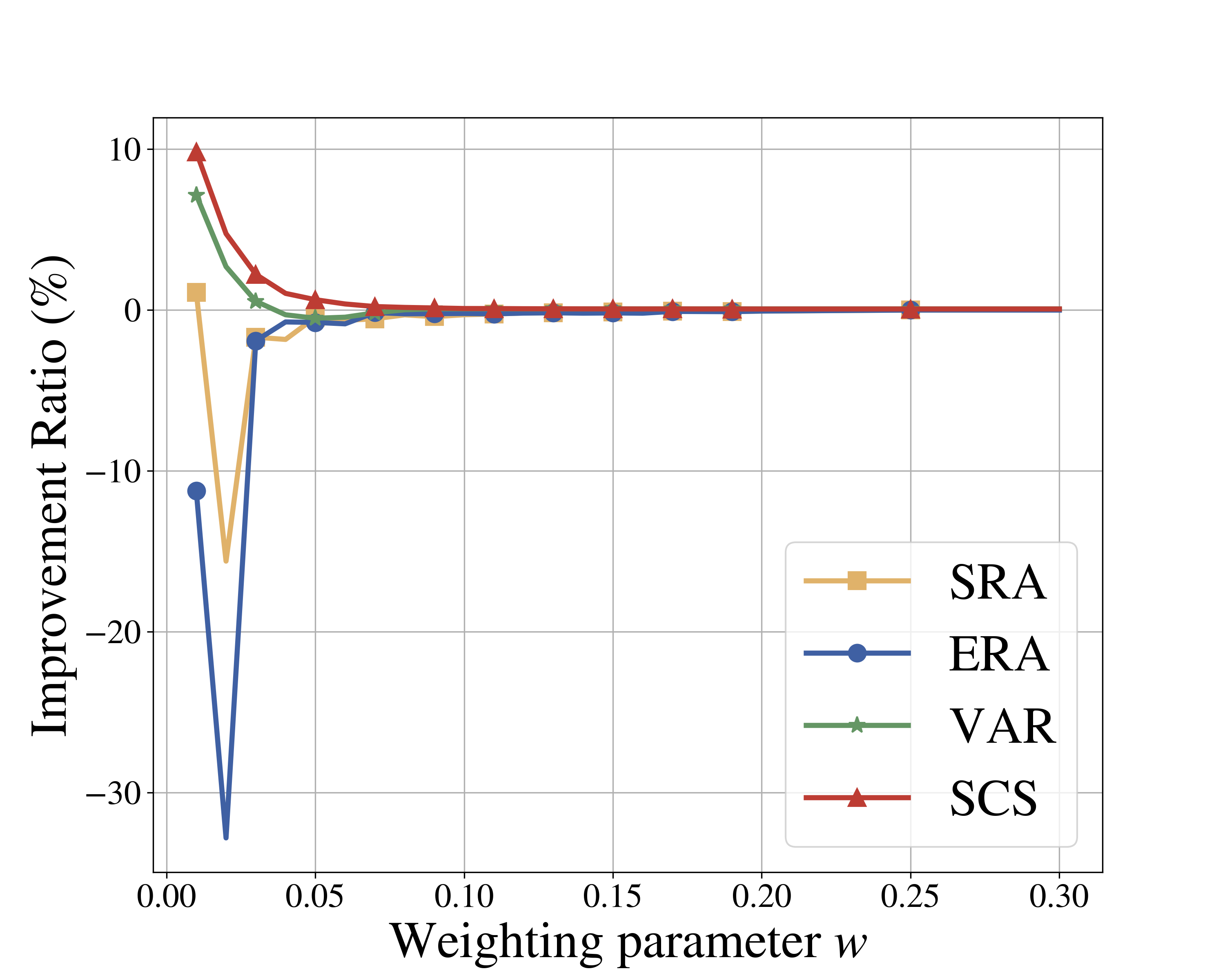}
        \caption{Log-normal ($N=50$)}
        \label{fig:exp:syn:LN:N50}
    \end{subfigure}
    \hfill
    \begin{subfigure}[b]{0.24\textwidth}
        \centering
         \includegraphics[width=\linewidth]{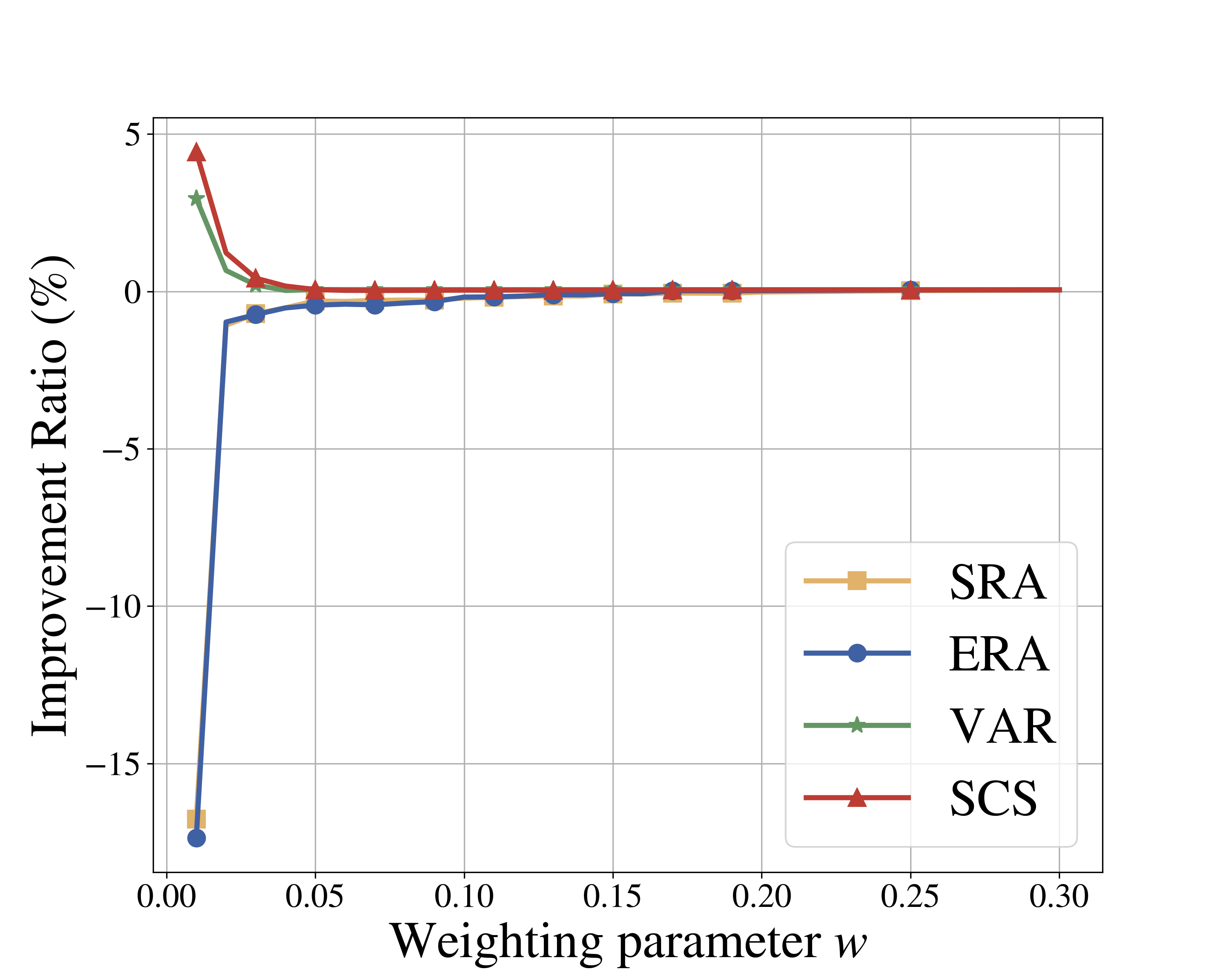}
        \caption{Log-normal ($N=100$)}
        \label{fig:exp:syn:LN:N100}
    \end{subfigure}     
    \caption{Synthetic simulation: improvement ratio. 
    }
    \label{fig:exp:syn}
\end{figure*}

\subsection{Synthetic Trace Simulation}
\label{sim:syn}

We evaluate content delivery latency across varying system scales, with the number of users $N$ ranging from $5$ to $100$. We use two common distributions: truncated normal and log-normal\footnote{Log-normal distributions have been shown to accurately model real-world network traffic, especially in large-scale data center systems~ \cite{gualandi2018call,zhou2021learning}.}. The detailed parameter settings are as follows:

\begin{itemize}
    \item \textit{Truncated Normal}: The mean is uniformly sampled from $10–200$\,ms, and the standard deviation from $1–100$\,ms. Latency values are truncated between 0 and $1000$\,ms to reflect practical bounds.
    
    \item \textit{Log-Normal}: The mean is uniformly sampled from $0.1–0.2$\,s, and the standard deviation from $1–100$\,ms. To ensure finite support, latency values are clipped at $10$\,s.
\end{itemize}

\subsubsection{Impact of parameter $T$ in SCS}
Our Proposition 2 indicates that theoretically, the precision of ADMM iterations, $\epsilon_k$, needs to decay at rate $\mathcal{O}(1/k^p)$ with $p>2$ to ensure the convergence to the optimal objective value. However, in practice, such a choice can be too conservative and unnecessarily increase the computation time. We would like to evaluate the impact of early termination in ADMM by executing a finite $T$  number of iterations on the convergence of our proposed algorithm. Here, the optimal objective value $F^\star$ is obtained using an offline Monte Carlo simulation (see \cite{kim2015guide}). As shown in Fig.~\ref{fig:exp:syn:convgence}, the proposed algorithm consistently converges to the optimal objective value across both latency distributions and a wide range of user scales ($N = 5 $ to $100$), demonstrating strong robustness and scalability. The designed SCS algorithm remains effective with a finite number of ADMM iterations. When \( T = 10 \) or \( 20 \), the algorithm converges with negligible optimality gaps. Even with \( T = 5 \), it achieves near-optimal performance in most cases, with only slight degradation under large user scales (see Fig. ~\ref{fig:exp:syn:convgence}(c) and (d)), showcasing its practicality under constrained computational budgets. Moreover, larger values of $T$ (i.e., a more accurate solution for each content delivery) lead to faster convergence across all settings, indicating a favorable trade-off between computational cost and performance in practice.

\subsubsection{Improvement ratio}
Next, we evaluate the performance of the proposed SCS algorithm with $T=10$ compared with the baseline methods SRA, ERA, and VAR, as illustrated in Fig. \ref{fig:exp:syn}. From Fig. \ref{fig:exp:syn:TG:N20}, we can observe that SCS outperforms all baseline algorithms when the latency follows the truncated normal distribution. Specifically, SCS achieves up to more than $80\%$ improvement ratio when the weighting parameter $w$ is small while SRA and ERA have relatively small improvement. Notably, VAR can even degrade synchronization performance, highlighting that merely minimizing latency variance across users cannot guarantee effective synchronization. For the Log-normal latency distribution, as shown in Fig. \ref{fig:exp:syn:LN:N50} and \ref{fig:exp:syn:LN:N100}, SCS again demonstrates superior performance, particularly when $w$ is small. It significantly outperforms both SRA and ERA, and maintains an advantage over VAR.

 \wang{To further illustrate the role of the weighting parameter $w$, we report the average max-min delay gap and the average violation probability under different values of $w$, as shown in Fig.~\ref{fig:exp:trade-off}. We consider $N=5$ users with log-normal latency distributions and define a violation event as $x_k(n)+d_k(n)>250$ ms. A smaller $w$ prioritizes synchronization and yields a smaller average max-min delay gap, whereas a larger $w$ reduces the average violation probability at the cost of a larger delay discrepancy. These results confirm that $w$ provides a tunable tradeoff between synchronization performance and delay-constraint satisfaction.}
 \sndrev{
 Accordingly, synchronization-critical applications, such as collaborative MR, may favor a smaller $w$, whereas applications prioritizing compliance with latency requirements, such as multi-party video conferencing, may favor a larger $w$.
 In practice, $w$ can be selected using representative validation traces to meet application-specific requirements. For instance, under the setting in Fig.~\ref{fig:exp:trade-off}, a target average violation probability of around 3\% at the $250$ ms threshold suggests $w \approx 0.1$. More generally, the operator may choose the smallest feasible $w$ under a latency constraint, or the largest feasible $w$ under a synchronization-gap requirement to reduce latency.}
\begin{figure}
    \centering
    \includegraphics[width=0.7\linewidth]{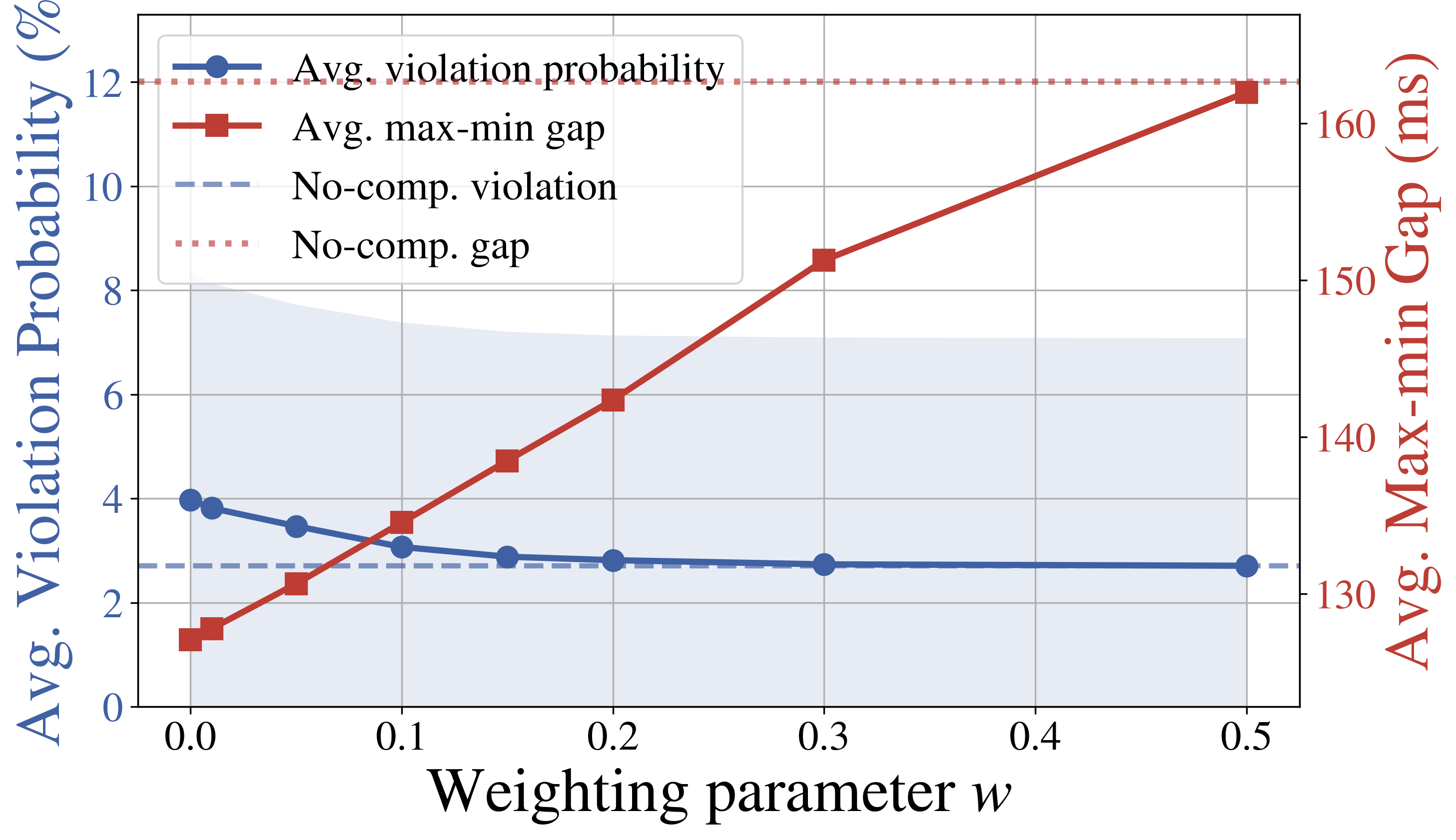}
    \caption{\wang{Synchronization–violation tradeoff. The shaded region shows the average user violation probability $\pm$ one standard deviation. Horizontal lines denote no-compensation baselines.}}
    \label{fig:exp:trade-off}
\end{figure}

\subsubsection{Convergence rate comparison}
We compare the proposed SCS algorithm, based on the SPP framework, with the SSGD method. Our experiments consider $N = 5$ users experiencing latencies modeled as latency-truncated normal distributions, under varying mean latency $\boldsymbol{\mu}$ and standard deviation $\boldsymbol{\sigma}$ patterns, detailed in Table~\ref{table:converg:param}:

\begin{table}[H]
\centering
\caption{Latency parameters configurations.}
\label{table:converg:param}
\begin{tabular}{lcc}
\toprule
Pattern & Mean vector $\boldsymbol{\mu}$\,[ms] & Std.~dev.\ vector $\boldsymbol{\sigma}$\,[ms] \\
\midrule
Low $\boldsymbol{\mu}$, low $\boldsymbol{\sigma}$  & $[2,\;4,\;6,\;8,\;10]$        & $[0.5,\;1,\;1.5,\;2,\;2.5]$  \\[2pt]
High $\boldsymbol{\mu}$, low $\boldsymbol{\sigma}$  & $[10,\;20,\;30,\;40,\;50]$    & $[0.5,\;1,\;1.5,\;2,\;2.5]$  \\[2pt]
High $\boldsymbol{\mu}$, high $\boldsymbol{\sigma}$  & $[10,\;20,\;30,\;40,\;50]$    & $[5,\;10,\;15,\;20,\;25]$    \\
\bottomrule
\end{tabular}
\end{table}

\begin{figure}
     \begin{subfigure}{0.24\textwidth}
        \centering
        \includegraphics[width=\linewidth]{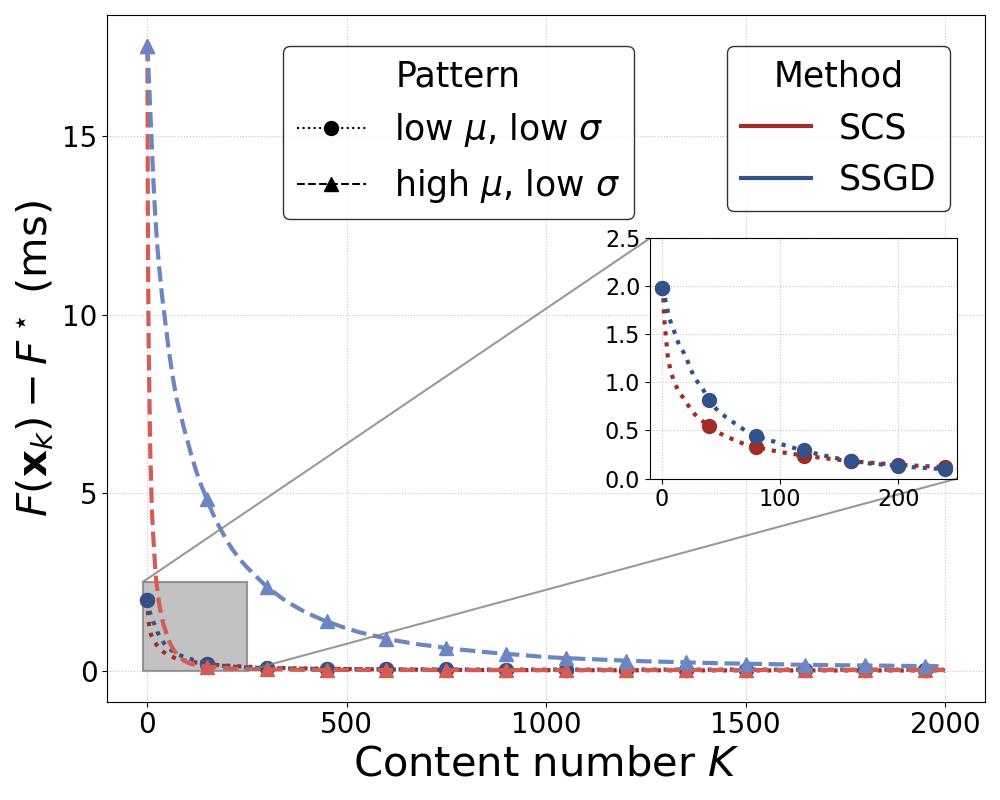}
        \caption{Effect of different mean latency on convergence}
        \label{fig:exp:subgradient:mean}
    \end{subfigure}
    \hfill
    \begin{subfigure}{0.24\textwidth}
        \centering
        \includegraphics[width=\linewidth]{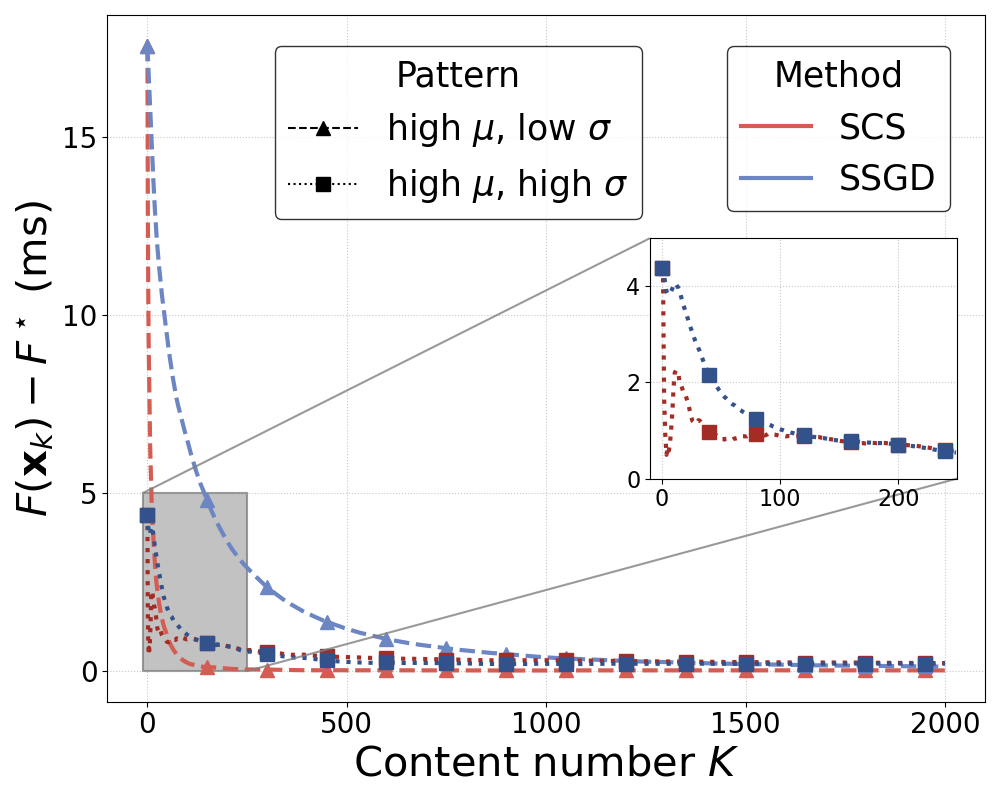}
        \caption{Effect of different latency variance on convergence}
        \label{fig:exp:subgradient:variance}
    \end{subfigure}
    \caption{Synthetic simulation: convergence of SCS vs. SSGD under varying latency distribution patterns.
    }
    \label{fig:exp:ssgd}
\end{figure}

We first evaluate the impact of the mean latency vector $\boldsymbol{\mu}$ by fixing the standard deviation $\boldsymbol{\sigma}$ (Fig.~\ref{fig:exp:subgradient:mean}). As $\boldsymbol{\mu}$ increases, the latency disparities among the users become more significant. Although both SCS and SSGD ultimately converge to the optimal value, SCS consistently requires fewer content deliveries to achieve near-optimality. Notably, when the mean latency differs from $40$ ms, SCS reaches near-optimality performance in about $250$ deliveries, whereas SSGD necessitates approximately $800$, demonstrating SCS’s superior efficiency under high mean latency. Next, we compare performance under two variance conditions: high $\boldsymbol{\mu}$ with low $\boldsymbol{\sigma}$ versus high $\boldsymbol{\mu}$ with high $\sigma$  (Fig.~\ref{fig:exp:subgradient:variance}). As $\boldsymbol{\sigma}$ increases, both algorithms experience slower convergence due to the increased variance in latency samples. In the setting with low $\boldsymbol{\sigma}$, indicating a more stable network with low jitter, SCS converges in fewer than 150 deliveries, whereas SSGD requires more than $500$ round contents. Overall, SCS consistently exhibits superior convergence speed to SSGD across a diverse range of latency patterns.

\wang{
\subsubsection{Impact of parameter $\rho$}

We evaluate the sensitivity of SCS to the ADMM penalty parameter $\rho$ across two representative settings: truncated normal latencies with $N=5$ users and log-normal latencies with $N=50$ users. As shown in Fig.~\ref{fig:exp:parameter}, while slight differences in convergence speed are observed, all tested values of $\rho$ converge to near-optimal performance in both settings, demonstrating that the algorithm is robust to the choice of $\rho$. These results suggest that, in practice, $\rho$ can be selected within a reasonable range without significantly affecting synchronization performance.}

\begin{figure}
     \begin{subfigure}{0.24\textwidth}
        \centering
        \includegraphics[width=\linewidth]{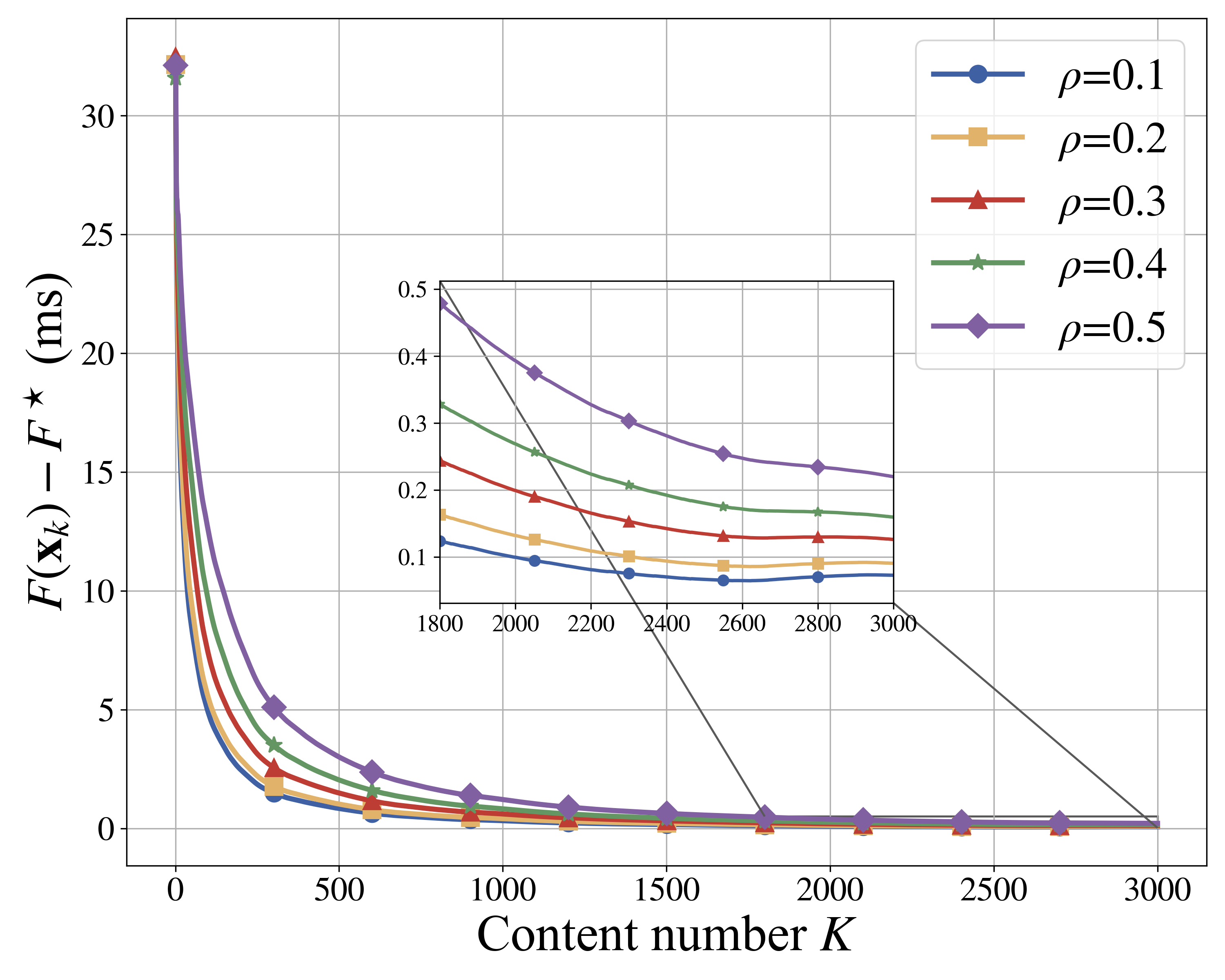}
        \caption{Truncated normal ($N$=5)}
        \label{fig:exp:parameter:N5}
    \end{subfigure}
    \hfill
    \begin{subfigure}{0.24\textwidth}
        \centering
        \includegraphics[width=\linewidth]{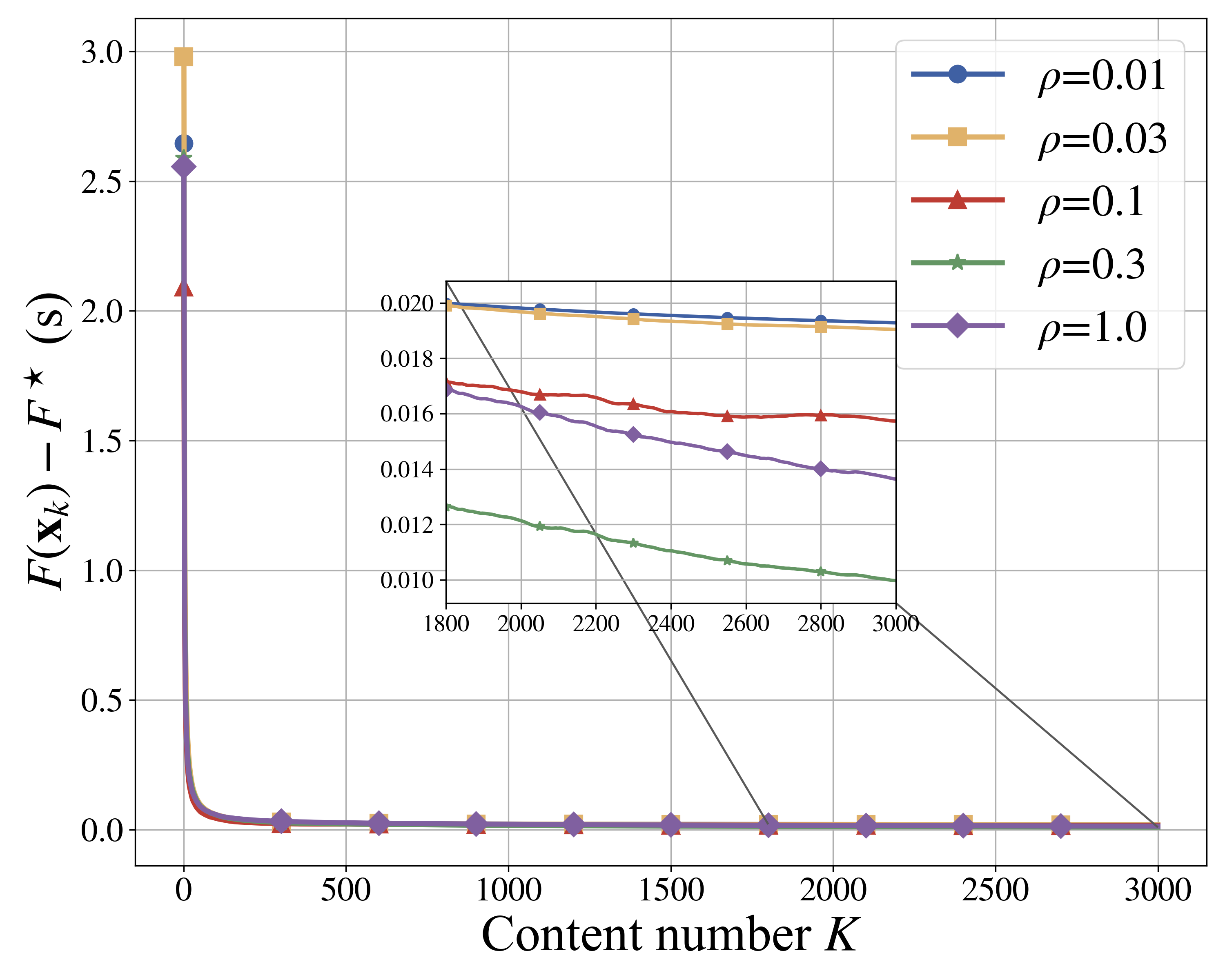}
        \caption{Log-normal ($N$=50)}
        \label{fig:exp:parameter:N50}
    \end{subfigure}
    \caption{\wang{Impact of parameter $\rho$ on the convergence of SCS.}}
    \label{fig:exp:parameter}
\end{figure}

\subsection{Real-world Trace Simulation}
In this subsection, we demonstrate the effectiveness of the proposed algorithm SCS using two real-world latency datasets.  First, we use the network traces from \cite{zhu2016network}, which collect the round trip times (RTTs) between $99$ nodes in the Seattle network over $688$ distinct time periods. We repeated the use of one trace to simulate the same network condition to let the nodes be 100.
\begin{figure*}[htbp]
    \begin{subfigure}[b]{0.24\textwidth}
        \centering
        \includegraphics[width=\linewidth]{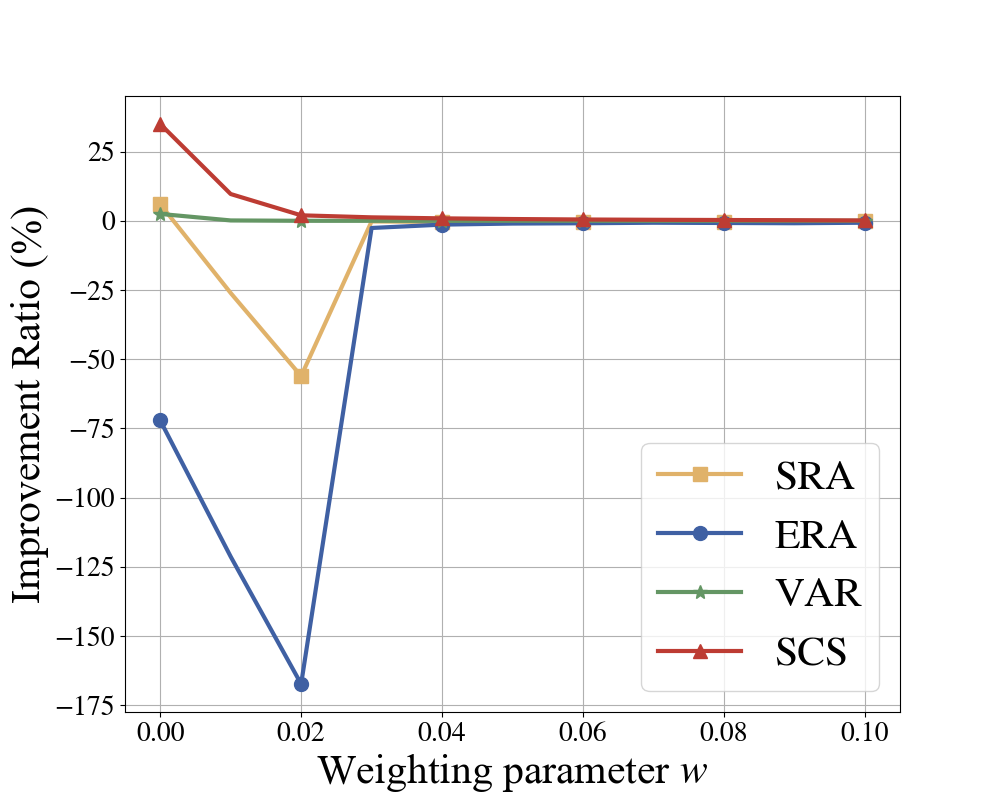}
        \caption{\centering Latency traces in Seattle \newline($N = 50$)}
        \label{fig:exp:real:Network:N50}
    \end{subfigure}
    \hfill
    \begin{subfigure}[b]{0.24\textwidth}
        \centering
        \includegraphics[width=\linewidth]{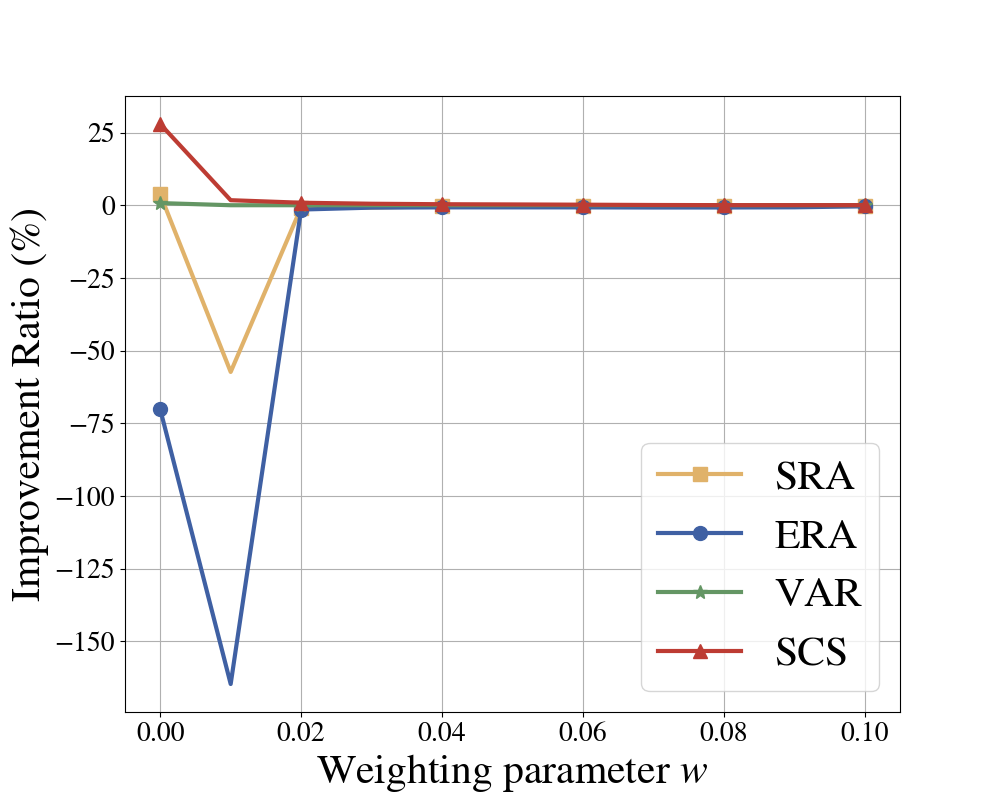}
        \caption{\centering Latency traces in Seattle \newline($N = 100$)}
        \label{fig:exp:real:Network:N100}
    \end{subfigure}
    \hfill
    \begin{subfigure}[b]{0.24\textwidth}
        \centering
        \includegraphics[width=\linewidth]{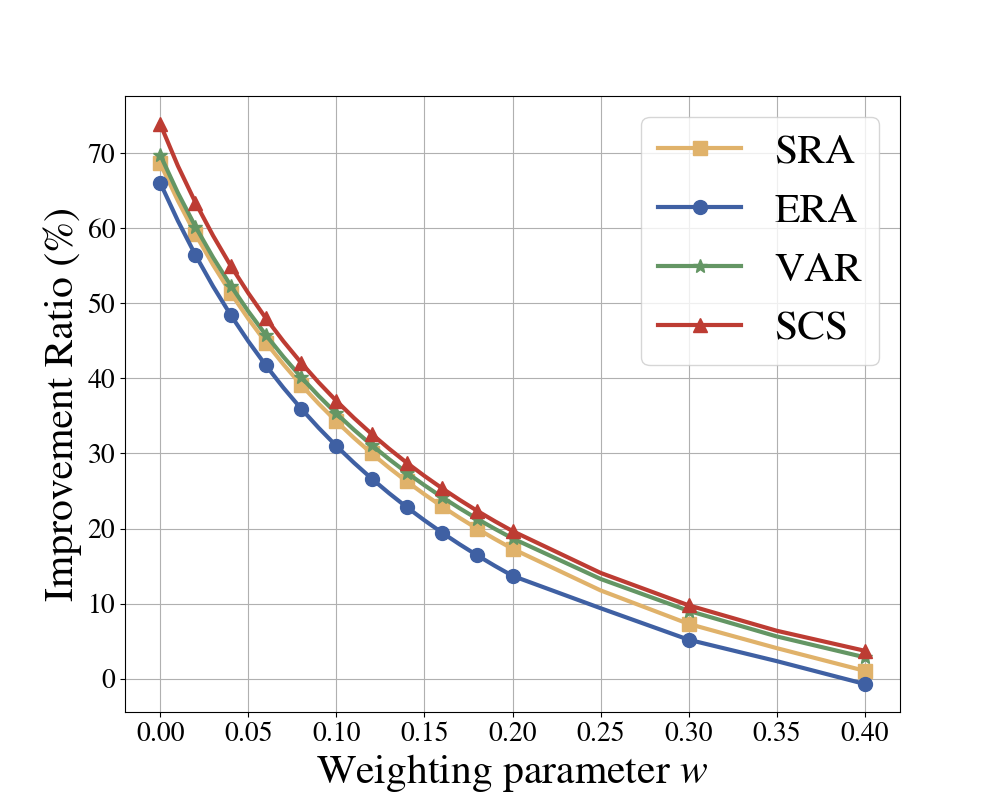}
        \caption{\centering MR-based firefighter training \newline($N = 5$)}
        \label{fig:exp:real:MR:N5}
    \end{subfigure}
    \hfill
    \begin{subfigure}[b]{0.24\textwidth}
        \centering
        \includegraphics[width=\linewidth]{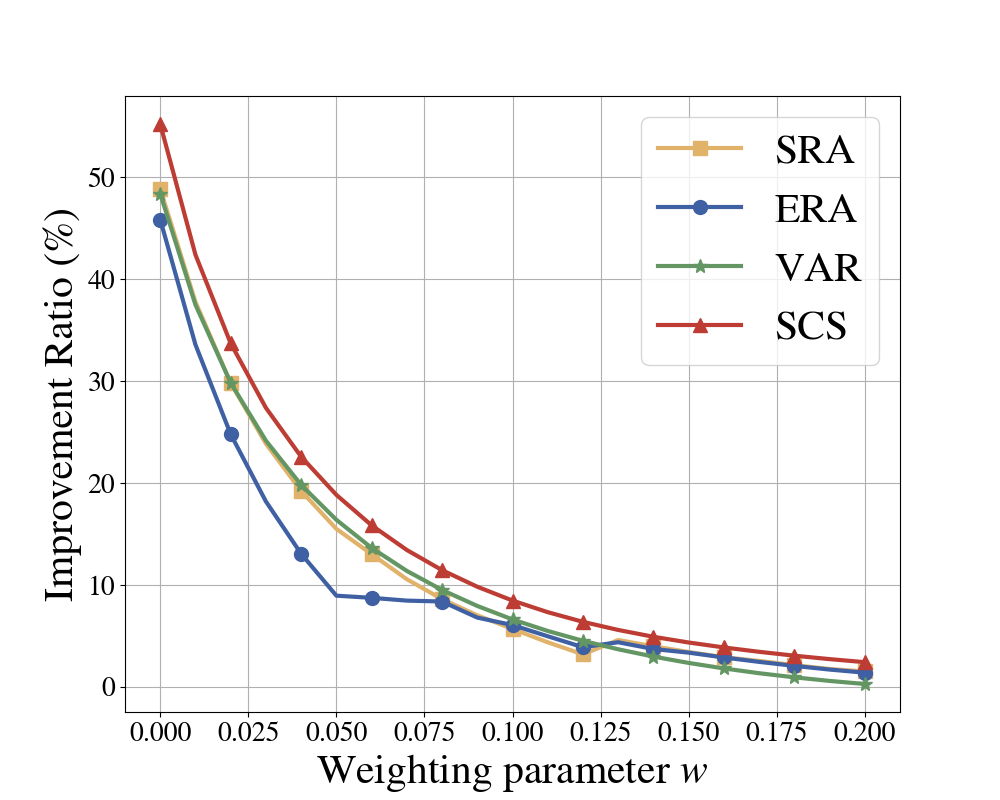}
        \caption{\centering MR-based firefighter training \newline ($N = 20$)}
        \label{fig:exp:real:MR:N20}
    \end{subfigure}
     \hfill
    \caption{Real-world trace simulation: improvement ratio.}
    \label{fig:exp:realtrace}
\end{figure*}

\subsubsection{Improvement ratio} 
 As shown in Fig. \ref{fig:exp:real:Network:N50} and \ref{fig:exp:real:Network:N100}, we observe that the proposed SCS consistently exhibits a larger improvement ratio than baseline algorithms. Specifically, when $w = 0.01$ and $N=50$ (see Fig. \ref{fig:exp:real:Network:N50}), SCS achieves an improvement ratio of around 10\%, while SRA and ERA adversely impact performance, resulting in undesirable experiences for users. VAR, in contrast, shows no meaningful improvement in large-scale network synchronization. Additionally, we examine the performance of SCS on the dataset from an MR-based firefighter training platform, which records real transmission latency from the server to multiple participants \cite{dong2023collaborative}. As shown in Fig. \ref{fig:exp:real:MR:N5} and \ref{fig:exp:real:MR:N20}, the improvement ratio of SCS is higher than SRA, ERA, and VAR across the weighting parameter $w$ from $0$ to $0.2$, This demonstrates the effectiveness of the proposed algorithm, especially when synchronization performance is prioritized. These results highlight the robustness and effectiveness of SCS, particularly in scenarios where precise synchronization is critical.

\subsubsection{Computational time}
Moreover, we evaluate the computational time of the proposed algorithm SCS. Fig. \ref{fig:exp:time:network} and Fig. \ref{fig:exp:time:MR}  illustrate the computational time per content delivery compared against the toolbox CVXPY. We can see that the computational time of SCS is significantly less than that of CVXPY. Specifically, the computation time for CVXPY exceeds $5$ ms, while our algorithm requires usually less than $2$ ms even when there are 50 users. This illustrates that the proposed SCS is time-efficient, making it promising for real-time content delivery in practical systems.

\begin{figure}   %
     \begin{subfigure}{0.24\textwidth}
        \centering
        \includegraphics[width=\linewidth]{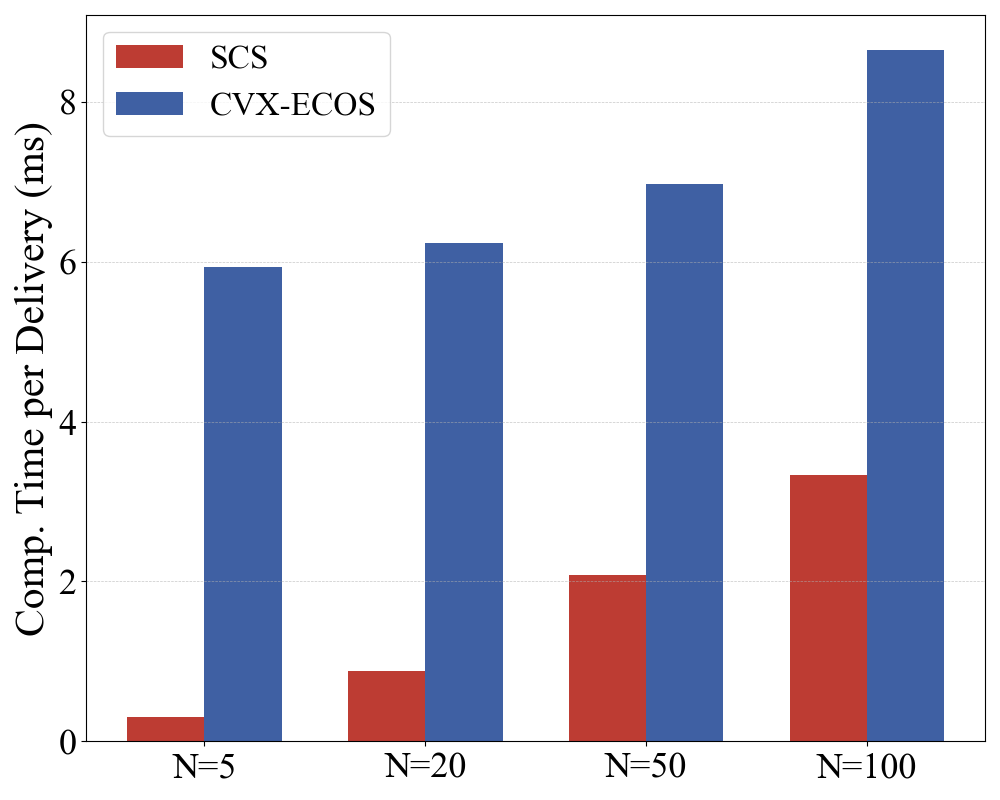}
        \caption{Latency traces in Seattle}
        \label{fig:exp:time:network}
    \end{subfigure}
    \hfill
    \begin{subfigure}{0.24\textwidth}
        \centering
        \includegraphics[width=\linewidth]{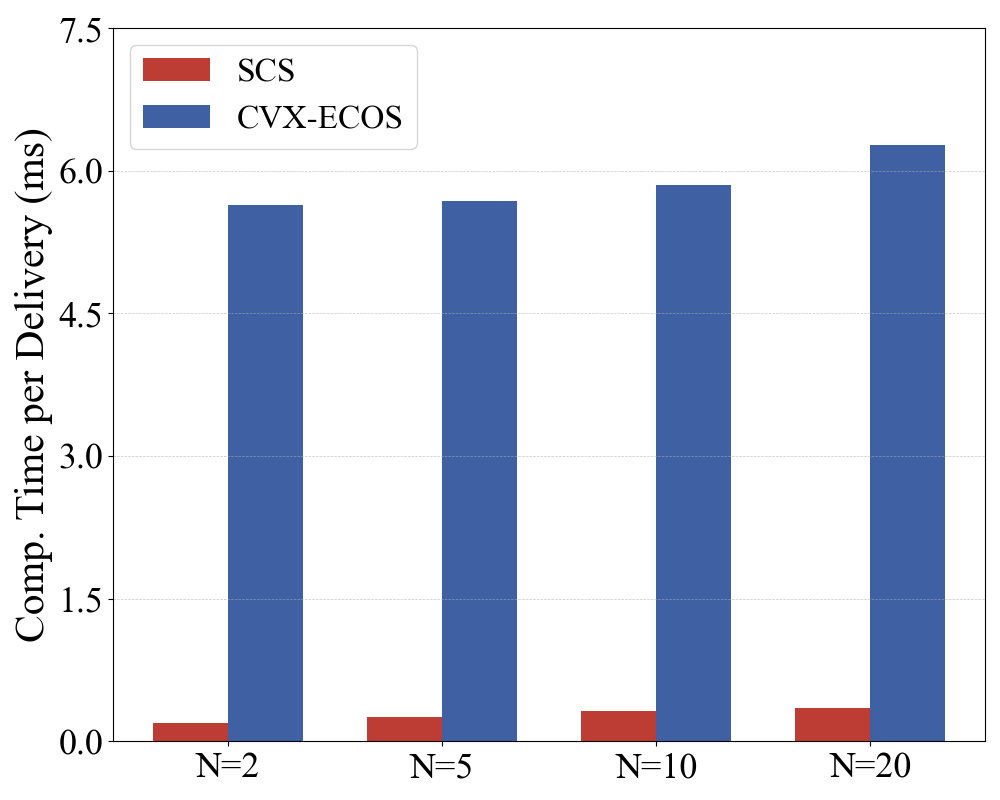}
        \caption{MR-based firefighter training}
        \label{fig:exp:time:MR}
    \end{subfigure}
    \caption{Real-world trace simulation: computational time.}
    \label{fig:exp:runtime}
\end{figure}

\subsubsection{Adaptability to different cost functions} 
We study the practical impact of varying cost functions in the synchronization tasks. We employ five representative network traces and set the weighting parameter $w = 0.1$ to ensure a consistent experimental environment. The proposed SCS algorithm is compared against baseline algorithms across three common cost functions $c_n(x)$:  linear $(x)$, quadratic $(x^2)$, exponential ($e^{0.1x}-1$). As depicted in Fig.~\ref{fig:VaryingCostFunctions}, SCS consistently achieves a higher improvement ratio across all cost functions. Particularly, under the exponential function, SCS achieves over $70\%$ improvement, significantly outperforming SRA, which attains approximately $20\%$ improvement. These results highlight the strong adaptability of SCS across diverse latency cost models in practical synchronization scenarios.

\begin{figure}[h]
    \centering
    \includegraphics[width=0.7\linewidth]{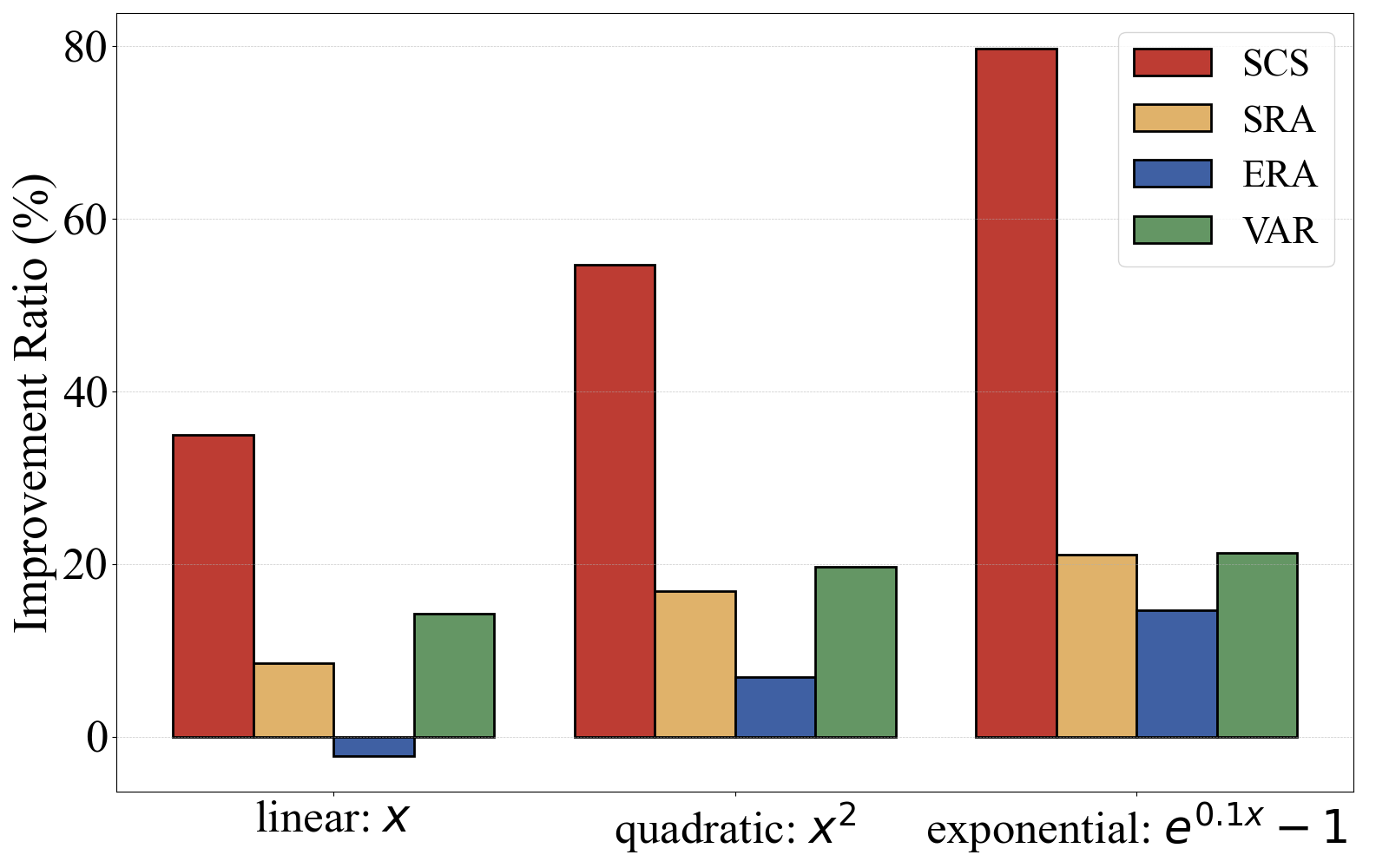}
    \caption{Real-world trace simulation: adaptability to different cost functions. }
    \label{fig:VaryingCostFunctions}
\end{figure}

\subsection{Dropping Case Simulation}
\label{sim:dropping_case}
To evaluate the performance of the proposed D-SCS algorithm under dropping case, we introduce an i.i.d.\ packet dropping rate with a fixed probability $q(n) = q$ for all users. The dropping probability $q$ takes values from the set $ \{0,\,0.05,\,0.1,\,0.2\}$, simulating increasing challenging network environments. We evaluate the algorithm in two representative scenarios: $N=5$ users with truncated-normal latencies and $N=20$ users with log-normal latencies. Fig.~\ref{fig:exp:dropping:conv} shows the objective value versus the number of contents~$K$ in which dashed lines indicate the optimal function value corresponding to each dropping rate. In both scenarios, D-SCS algorithm demonstrates robust convergence behavior, maintaining near-optimal performance even under severe packet loss conditions (e.g., $q=0.2$). These results underscore the resilience and effectiveness of the proposed algorithm in maintaining synchronization in realistic network conditions with packet drops.

\begin{figure}[t!]
     \begin{subfigure}{0.24\textwidth}
        \centering
        \includegraphics[width=\linewidth]{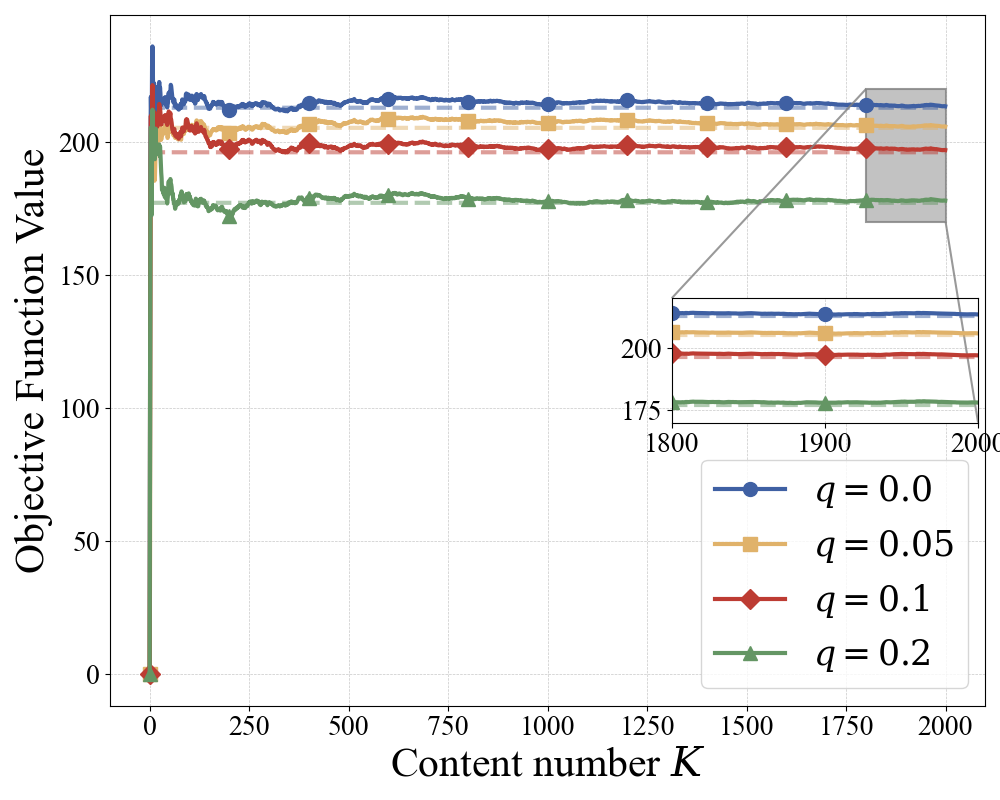}
        \caption{Truncated normal ($N$=5)}
        \label{fig:exp:dropping:conv:1}
    \end{subfigure}
    \hfill
    \begin{subfigure}{0.24\textwidth}
        \centering
        \includegraphics[width=\linewidth]{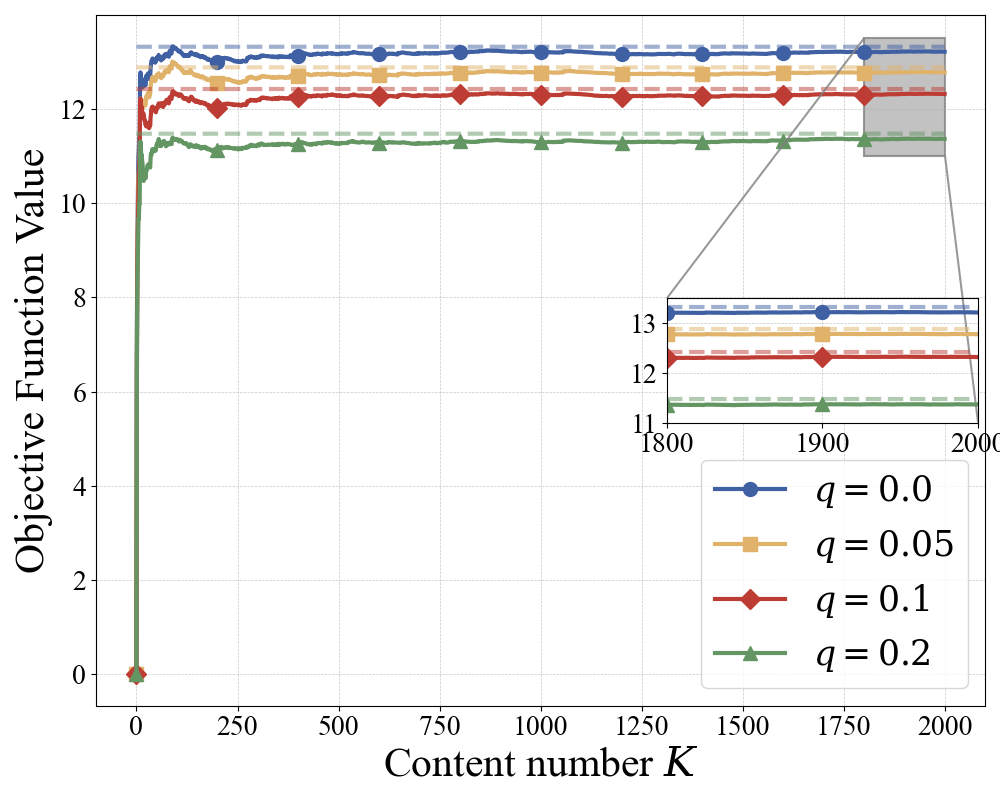}
        \caption{Log-normal ($N$=20)}
        \label{fig:exp:dropping:conv:2}
    \end{subfigure}
    \caption{Dropping case simulation: convergence under varying dropping rates. Dashed lines denote the optimal function value.}
    \label{fig:exp:dropping:conv}
\end{figure}

\section{Conclusion}

In this paper, we considered the problem of real-time synchronized scheduling for collaborative content delivery from a server to a group of users. We introduced latency compensation to control the server’s delivery time for each user. Crucially, our framework is clock-free, relies only on simple acknowledgment feedback without requiring a synchronized clock or client-side timestamps. Our goal is to minimize the discrepancy between maximum and minimum latencies while keeping the overall latency costs experienced by all users low. Due to the randomness of latency, we formulated this as a stochastic optimization problem involving a non-smooth component from the synchronization. To tackle this challenge, we utilized SPP, which only requires historical latency samples. We explored the inherent problem structure and leveraged the ADMM to address the challenges posed by constraints and the non-smooth component, resulting in only the computational complexity of $\mathcal{O}(N\log N)$ per iteration. We further demonstrate that our proposed algorithm converges to the optimal objective value at the rate of $\mathcal{O}(1/\sqrt{K})$. Additionally, we extended our design to handle packet drops during delivery, confirming its robustness in degraded network conditions. Finally, our simulations evaluated the effectiveness of the proposed algorithm against baselines. \wang{An interesting future direction is to extend the framework to dynamic user scenarios where users may join or leave, which would require modeling time-varying active user set and analyzing convergence under changing problem dimensions.}

\appendices

\section{Proper and Convex Objective Function}
\label{appendix:convex}
\begin{claim}
\label{claim:convex}
   $f(\, \cdot \,;\mb{d})$ and $F$ are proper and convex functions.
\end{claim}

\begin{proof}
    Given any realization of the $N$-dimensional random vector $\mathbf{d}$, the function $\max_{n \in [N]}(x_n + d_n)$ is convex w.r.t. $\mathbf{x}$, $\min_{n \in [N]}(x_n + d_n)$ is concave w.r.t. $\mathbf{x}$, and $w \sum_{n \in [N]} c_n(x_n + d_n)$ is convex w.r.t. $\mathbf{x}$. Therefore, the function $f(\mathbf{x}; \mathbf{d})$ is convex w.r.t. $\mathbf{x}$ for any fixed $\mathbf{d}$. Recall $F(\mathbf{x}) \triangleq \mathbb{E} [f(\mathbf{x}; \mathbf{d})]$. For any $\lambda \in [0, 1]$, we have
    \begin{align}
        F(\lambda \mathbf{x}_1 + (1-\lambda) \mathbf{x}_2) 
        =& \ \mathbb{E} \left[ f(\lambda \mathbf{x}_1 + (1-\lambda) \mathbf{x}_2; \mathbf{d}) \right] \nonumber \\
        \leq & \ \mathbb{E} \left[ \lambda f(\mathbf{x}_1; \mathbf{d}) + (1-\lambda) f(\mathbf{x}_2; \mathbf{d}) \right] \nonumber \\
        = & \ \lambda \mathbb{E} [f(\mathbf{x}_1; \mathbf{d})] + (1-\lambda) \mathbb{E}_{\mathbf{d}} [f(\mathbf{x}_2; \mathbf{d})] \nonumber \\
        = & \ \lambda F(\mathbf{x}_1) + (1-\lambda) F(\mathbf{x}_2), \nonumber
    \end{align}
    where the second step follows from the convexity of $f(\mb{x};\mb{d})$. Thus, $F$ is a convex function.

    Next, we show that the function $f(\cdot;\mb{d})$ is proper. 
    In addition, given that $\mb{x}$ and $\mb{d}$ are non-negative, it follows that $f(\mb{x};\mb{d}) \geq 0$ for all $\mb{x}\succeq\mb{0}$. Moreover, when $\mb{x} = \mb{0}$, $f(\mb{0};\mb{d}) < +\infty$ for the finite $\mb{d}$. Hence, the function $f(\cdot;\mb{d})$ is proper and consequently, $F$ is also proper. 
\end{proof}

\section{Proof of Lemma~\ref{lem:order-preserve}}\label{pf:lem-1}

For simplicity, we denote $y_n^{\star}$ as $y^\star(n)$ for any $n \in [N]$ in this proof. We want to show that $y_m^{\star}\geq y_n^{\star}$ for $m>n$. We will prove this by contradiction. Suppose $y_m^\star < y_n^\star$ when $m>n$. According to the Lemma condition, there are only two cases to consider: $x_m>x_n$ and $x_m=x_n$. 

First, if $x_m>x_n$, we will show that swapping $y_m^\star$ and $y_n^\star$ will reduce the objective value, contradicting the optimality of $\mb{y}$. Second, if $x_m=x_n$, the swap will generate 
two optimal solutions, contradicting the uniqueness of $\mb{y}^\star$. Since $\widetilde{f}$ is permutation invariant, we need to focus only on the quadratic term. By computing the difference in the objective value before and after swapping $y_m^\star$ and $y_n^\star$, we have
\begin{align}
&  (y_m^\star - x_m)^2 +  (y_n^\star - x_n)^2 \nonumber\\
= \ &  (y_m^\star)^2 + (y_n^\star)^2 + x_m^2 + x_n^2- 2(y_m^\star x_m + y_n^\star x_n) \nonumber \\
{\geq} \ &  (y_m^\star)^2 + (y_n^\star)^2  + x_m^2 + x_n^2 - 2(y_m^\star x_n + y_n^\star x_m) \label{eqn:lem:order-preserve}\\
= \ & (y_n^\star - x_m)^2 + (y_m^\star - x_n)^2, \nonumber
    \end{align}
    where the inequality holds since
    \begin{align*}
        \begin{split}
             & (y_m^\star x_m + y_n^
             \star x_n ) - (y_m^\star x_n + y_n^\star x_m)\\
   = \ & y_m^\star (x_m - x_n) + y_n^\star (x_n - x_m) \\
   = \ &(y_m^\star - y_n^\star)(x_m - x_n) \leq 0.
        \end{split}
    \end{align*}
If $x_m>x_n$, then \eqref{eqn:lem:order-preserve} holds strictly and thus swapping $y_m^\star$ and $y_n^\star$ yields a smaller objective value, contradicting the optimality of $\mb{y}^\star$. If $x_m=x_n$, then \eqref{eqn:lem:order-preserve} becomes equality and swapping $y_m^\star$ and $y_n^\star$ does not change the objective value but yields another optimal solution, contradicting $\mb{y}^\star$'s uniqueness.

\section{Proof of lemma \ref{lemma:subgrad}}
\label{Appendix:subdifferential}
     Define $\mc{I}_{\max} \triangleq \arg\max_{n \in [N]} \{x(n) + d(n)\}$ and $\mc{I}_{\min} \triangleq \arg\min_{n \in [N]} \{x(n) + d(n)\}$. Then the subdifferential of the non-smooth part is 
    \begin{align}
           &\partial \left(\max_n(x(n)+d(n)) - \min_n(x(n)+d(n)) \right) \nonumber \\ 
           &= \operatorname{conv}\{ e_i : i \in \mc{I}_{\max} \} - \operatorname{conv}\{ e_j : j \in \mc{I}_{\min} \} , \nonumber
    \end{align}
where  $e_i$ is the standard basis vector with the $i^{th}$ component being $1$ and $0$ for all other components, and \text{conv} represents convex hull (see \cite{beck2017first}, Ch. 3). Any subgradient $\zeta \in \partial \left(\max_{n }(x(n)+d(n)) - \min_{n}(x(n)+d(n)) \right) $ satisfies

\begin{equation*}
    \zeta = \sum_{ i \in \mc{I}_{\max}} \alpha_ie_i - \sum_{ j \in \mc{I}_{\min}} \beta_je_j 
\end{equation*}
with $\sum_{ i \in \mc{I}_{\max}}\alpha_i = \sum_{ j \in \mc{I}_{\min}}\beta_j = 1$ and $\alpha_i,\beta_j \geq0 ;\forall  i,j$. Using the orthogonality of the standard basis, we have

\begin{align}
    \|\zeta\|^2 &= \sum_{i \in \mc{I}_{\max}} \alpha_i^2 + \sum_{j \in \mc{I}_{\min}} \beta_j^2 \leq 2
    \label{lem:subg:nonsmooth}
\end{align}

In addition, since \( c_n(\cdot) \) is differentiable and continuous on $ \mathbb{R}_+ $, it follows that over any closed and bounded interval $ [0, \tau_{\max}] $, the derivative $ c_n'(x) $ is bounded. Hence, for any finite $ x(n) \geq 0 $ and $ x(n) + d(n) \in [0, \tau_{\max}] $, there exists a constant $ L_n \geq 0 $ such that
$$
c_n'(x(n) + d(n)) \leq L_n, \quad\forall n\in[N].
$$
This follows from the Extreme Value Theorem, since $ c_n'(\cdot) $, being continuous on a closed and bounded domain, must attain a finite maximum.

Define $L \triangleq \max_n L_{n}$, we have the gradient from the smooth component satisfying
\begin{align}
    \label{lem:subg:smooth}
   \Bigg \|w\sum_{n\in [N]} c'_n(x(n) + d(n)) \cdot e_n\Bigg \|^2 &\leq w ^2\sum_{n\in [N]} L_n^2 \nonumber \\
   &\leq w^2N L^2   
\end{align}

Combining ~\eqref{lem:subg:nonsmooth} and ~\eqref{lem:subg:smooth}, using the inequality $\|a + b\|^2 \leq 2\|a\|^2 + 2\|b\|^2$, taking expectation give the desired results:
\begin{align}
    \mathbb{E}\left[ \left\| f'(\mathbf{x}; \mathbf{d}) \right\|^2 \right]
&\leq 2 \cdot \mathbb{E}[\|\zeta\|^2] + 2N w^2 L^2  \nonumber \\
& \leq 4 +  2N w^2 L^2   
\end{align}

\bibliographystyle{IEEEtran}
\bibliography{refs}

\end{document}